\documentclass[letter,11pt]{article}

\usepackage{amsmath,amsthm,amssymb,euscript}
\usepackage{enumerate}
\usepackage{graphicx}
\usepackage{comment}
\usepackage{xspace}
\usepackage{multirow}
\usepackage{url}
\usepackage{cite}
\usepackage{rotating}
\usepackage[shortlabels]{enumitem}
\usepackage{color}
\usepackage{subfigure}
\usepackage[left=1in,right=1in,top=1in,bottom=1in]{geometry}
\usepackage{float}

\floatstyle{boxed}
\newfloat{pseudoalg}{tb}{palg}
\floatname{pseudoalg}{Algorithm}

\usepackage{ifthen}
\newcounter{Accumulate} 
\ifthenelse{\value{Accumulate} = 1}{
  \newwrite\accuwrite \immediate\openout\accuwrite=\jobname.acc
}{}
\usepackage{environ}
\makeatletter
\newenvironment{accumulate}{\Collect@Body\accuPrint}{}
\makeatother
\newcommand{\accuPrint}[1]{
 \ifthenelse{\value{Accumulate} = 0}{%
      #1
  }
  {
    \newtoks\prxxxm
    \prxxxm{#1}
    \immediate\write\accuwrite{\the\prxxxm}
  }
}
\newcommand{\ifaccumulating}[1]{%
  \ifthenelse{\value{Accumulate} = 1}{%
    #1%
  }{}%
}
\newcommand{\ifnoaccumulating}[1]{%
  \ifthenelse{\value{Accumulate} = 0}{%
    #1%
  }{}%
}
\newcommand{\accuprint}{%
  \ifthenelse{\value{Accumulate} = 1}{
    \immediate\closeout\accuwrite
    \input{\jobname.acc} %
  }{}
}

\def\scri#1{{\EuScript#1}}
\def\OO{{\cal O}}

\newcommand{\crg}{\mathrm{cr}}
\newcommand{\crd}{\mathrm{cr}}
\newcommand{\crdr}[1]{\mathrm{cr}_{#1}}

\newcommand{\ins}{\mathrm{ins}}

\def\MEI(#1,#2){\mathrm{MEI}(#1,#2)}
\def\rMEI(#1,#2){\mathrm{r\mbox{-}MEI}(#1,#2)}
\def\INS(#1,#2){\mathrm{\ins}(#1,#2)}
\def\rINS(#1,#2){\mathrm{r\mbox{-}\ins}(#1,#2)}

\newcommand{\poly}{\mathit{poly}}

\newcommand{\defend}{\hfill$\diamond$}
\let\apxmark\relax
\ifaccumulating{\def\apxmark{{\bf\hspace*{-1ex}*}\,}}

\newtheorem{thm}{Theorem}
\newtheorem{claim}[thm]{Claim}
\newtheorem{lem}[thm]{Lemma}
\newtheorem{cor}[thm]{Corollary}

\theoremstyle{definition}
\newtheorem{defn}[thm]{Definition}
\theoremstyle{remark}
\newtheorem{rem}[thm]{Remark}

\begin{document}
\ifaccumulating{\setcounter{page}{0}}

\title{\bf Inserting Multiple Edges into a Planar Graph}
\author{{Markus Chimani}\\
\emph{\small Theoretical Computer Science, University Osnabr\"uck, Germany}\\
{\tt\small  markus.chimani@uni-osnabrueck.de}\\[2ex]
{Petr Hlin\v{e}n\'y}\footnote{P.~Hlin\v{e}n\'y has been supported by the Czech Science Foundation project 14-03501S.}\\
\emph{\small Faculty of Informatics, Masaryk University Brno, Czech Republic}\\
{\tt\small  hlineny@fi.muni.cz}}
\date{}

\maketitle

\vspace{1cm}
\begin{abstract}
Let $G$ be a connected planar (but not yet embedded) graph and $F$ a set of additional edges not yet in $G$. The
{\em multiple edge insertion} problem (MEI) asks for a drawing of $G+F$ with the
minimum number of pairwise edge crossings, such that the subdrawing of $G$ is
plane. An optimal solution to this problem approximates the crossing number of the graph $G+F$.

Finding an exact solution to MEI is NP-hard for 
general $F$, but linear time solvable for the special case of $|F|=1$ (SODA~01, Algorithmica)
or when all of $F$ are incident to a new vertex (SODA~09).

The complexity for general $F$ but with constant $k=|F|$ was open, but algorithms
both with relative and absolute approximation guarantees have been presented
(SODA 11, ICALP~11). We show that the problem is fixed parameter tractable (FPT) in $k$ 
for biconnected $G$, or if the cut vertices of $G$ have degrees bounded by a constant.
We give the first exact algorithm for this problem;
it requires only $\OO(|V(G)|)$ time for any constant $k$.
%
\end{abstract}


\clearpage

\section{Introduction}

The crossing number $\crg(G)$ of a graph $G$ is the minimum number of pairwise 
edge crossings in a drawing of $G$ in the plane. Finding the crossing number of a 
graph is one of the most prominent combinatorial optimization problems in graph 
theory and is NP-hard already in very restricted cases, e.g., even when 
considering a planar graph with one added
edge~\cite{DBLP:journals/siamcomp/CabelloM13} (cf.~the MEI problem for $k=1$ later). 
The problem has been vividly investigated for over 60 years, but there is still surprisingly little 
known about it; see e.g.~\cite{survscha} for an extensive reference.
While in general, there exists a $c>1$ such that the crossing number cannot be
approximated within a factor $c$ in polynomial
time~\cite{DBLP:journals/dcg/Cabello13}, 
several approximation algorithms arose for special graph classes. 

For general graphs with bounded degree, there is an algorithm that approximates 
the quantity $n+\crg(G)$ instead, giving an approximation 
ratio of $\OO(\log^2 n)$~\cite{DBLP:journals/jcss/BhattL84,egs}.
A sublinear approximation factor of $\tilde{\OO}(n^{0.9})$ for $\crg(G)$ in 
the bounded-degree setting was given in an involved
algorithm~\cite{stoccrossing}. We know constant factor approximations for 
bounded-degree graphs that are embeddable in some higher 
surface~\cite{ghls,torusHS,surfaceApprox}, or that have a small set of graph 
elements whose deletion leaves a planar graph---removing and re-inserting these 
elements can give strong approximation bounds such 
as~\cite{HS06,DBLP:journals/algorithmica/CabelloM11,apex,sodamultiedge}. 

In this paper, we follow the latter idea and concentrate on the \emph{Multiple 
Edge Insertion} problem $\MEI(G,F)$, to be formally defined in the section hereafter. 
Intuitively, we are given a planar graph $G$, and ask for the best way (in terms 
of total crossing number) to planarly draw $G$ and insert a set of new edges $F$ 
into $G$ such that the final drawing of $G+F$
(i.e., of the graph including the new edges of $F$) restricted to $G$ remains planar. 

This problem is polynomial-time solvable for $|F|=1$~\cite{GMW05} and in the
case when all edges of~$F$ are incident to a common vertex~\cite{insertvertex}, 
but NP-hard for general $F$~\cite{zieglerTh}.
Moreover, an exact or at least approximate MEI solution constitutes 
an approximation for the crossing number of the graph $G+F$~\cite{apex}. 
Considering constant $k:=|F|$, there have been two different approximation 
approaches~\cite{sodamultiedge} and~\cite{CHicalp}; the former one directly targets the 
crossing number and achieves only a relative approximation guarantee for MEI; 
the latter one first specifically attains an approximation of MEI with only an
additive error term, and then uses \cite{apex} to deduce a crossing number 
approximation. While the former one is not directly practical, 
the latter algorithm~\cite{CHicalp} in fact turns out to be one of the best choices to obtain strong 
upper bounds in practice~\cite{DBLP:journals/jgaa/ChimaniG12}.

In this paper, we show that for every constant $k$ and under mild
connectivity assumptions, there is an exact linear time algorithm, 
which has so far been an open problem even for $k=2$.
In terms of parameterized complexity, our algorithm is in FPT with 
the parameter~$k=|F|$.
\begin{thm}\label{thm:main}
Let $G$ be a planar connected graph on $n$ vertices, and $F$ a set of $k$
new edges (vertex pairs, in fact) where $k$ is a constant.  If $G$ is
biconnected, or the maximum degree of the cut vertices of $G$ is bounded by a
constant, then the problem $\MEI(G,F)$ is solvable in $\OO(n)$ time.
\end{thm}

We also mention that while the crossing number itself is in FPT w.r.t.\ the objective 
value~\cite{G,KR}, already a planar graph with one added
edge may have unbounded crossing number.

Both the aforementioned absolute MEI-approximation~\cite{CHicalp} and our new approach
can use \cite{apex} to obtain the same relative ratio for approximating the crossing number. 
However, our new approach does so without any additional additive term:

\begin{cor}
 Using the theorem relating an optimum $\MEI(G,F)$ solution 
to the crossing number of the graph~$G+F$~\cite{apex}, 
 Theorem~\ref{thm:main} gives a polynomial time
$k\Delta$-approximation for the crossing number of $G+F$ with constant $k=|F|$, 
where $G$ is a planar graph and $\Delta$ is its maximum degree.
\end{cor}



\paragraph{Organization.}
After formally defining our setting in the next section, 
we will concentrate on the still NP-hard problem \emph{Rigid MEI} in Section~\ref{sec:rigidMEI},
i.e., MEI under the restriction that the planar embedding of $G$
is fixed.
In Section~\ref{sec:nonrigid}, this algorithm is at the core of a dynamic programming over 
a decomposition tree of $G$, in order to obtain an FPT algorithm for the general MEI, i.e., when any planar embedding of $G$ is allowed.
The latter constitutes the result for Theorem~\ref{thm:main}.
\ifaccumulating{Due to restricted space, statements with proofs left for the appendix
are marked with an asterisk ~\apxmark$\!\!$.}


\section{Preliminaries}\label{sec:prelim}
\begin{accumulate}
\ifaccumulating{\subsection{Supplements for Section~\ref{sec:prelim}}}
\end{accumulate}

We use the standard terminology of graph theory. By default, we use the term 
{\em graph} to refer to a {loopless multigraph}, i.e., we allow parallel 
edges but no self-loops. If there is no danger of confusion, we denote an edge 
with the ends $u$ and $v$ chiefly by $uv$.

A \emph{drawing} of a graph $G=(V,E)$ is a mapping of the vertices $V$ to 
distinct points on a surface $\Sigma$, and of the edges $E$ to simple (polygonal) 
curves on $\Sigma$, connecting their respective end points but not containing any other 
vertex point. Unless explicitly specified, we will always assume $\Sigma$ to be 
the plane (or, equivalently, the sphere). A \emph{crossing} is a common point of 
two distinct edge curves, other than their common end point. Then, a drawing is 
\emph{plane} if there are no crossings. \emph{Plane embeddings} form equivalence 
classes over plane drawings, in that they only define the cyclic order of the 
edges around their incident vertices (and, if desired, the choice of the outer, 
infinite face). A \emph{planar} graph is one that allows a plane embedding. A 
\emph{plane} graph is an embedded graph, i.e., a planar graph together with a 
planar embedding.

Given a drawing $D$ of $G$, let $\crd(D)$ denote the number of 
pairwise edge crossings in $D$. The \emph{crossing number} problem asks for a 
drawing $D^\circ$ of a given graph $G$ with the least possible number 
$\crg(D^\circ)=:\crg(G)$.
By saying ``pairwise edge crossings'' we emphasize that we count a crossing point $x$ 
separately for every pair of edges meeting in $x$ (e.g., $\ell$ edges meeting in 
$x$ give ${\ell\choose2}$ crossings). 

\begin{defn}[Multiple edge insertion, MEI and rigid MEI]
\label{def:MEI}~\\
Consider a planar, connected graph $G$ and a set of edges 
(vertex pairs, in fact) $F$ not in~$E(G)$.
We denote by $G+F$ the graph obtained by adding $F$ to the edge set of~$G$.

Let $G_0$ be a planar embedding of $G$.
The \emph{rigid multiple edge insertion} problem $\rMEI(G_0,F)$ is to find 
a drawing $D$ of the graph $G+F$ with minimal $\crd(D)$ such that the restriction 
of $D$ to $G$ is the plane embedding $G_0$. 
The attained number of crossings is denoted by $\rINS(G_0,F)$.

The \emph{multiple edge insertion} problem $\MEI(G,F)$ is to find an
embedding $G_1$ of $G$ (together with the subsequent drawing $D$ as above),
for which $\rMEI(G_1,F)$ attains the minimum number of crossings.
The latter is denoted by $\INS(G,F)$.
\defend\end{defn}

Herein, we will also deal with the \emph{weighted} crossing number, i.e., we have 
edge weights $w\colon E(G)\to\mathbb{N}_+\cup\{\infty\}$, and a crossing between
two edges $e_1,e_2$ accounts for the amount of $w(e_1)\cdot w(e_2)$ in the above 
crossing functions.
Specially, for the MEI problem variants, we shall consider integer weights
on the edges of $G$ but not on $F$ (i.e., the weight on $F$ is always~$1$).
Although this is not a noteworthy strengthening of Theorem~\ref{thm:main} by
itself, the weights on $E(G)$ will be useful in the recursive processing of the
non-rigid case, cf.~Section~\ref{sec:nonrigid}.


Given a plane embedding $G_0$ of $G$, we define its (geometric) \emph{dual} 
$G_0^*$ as the embedded multigraph that has a (dual) vertex for each face in 
$G_0$; dual vertices are joined by a (dual) edge for each (primal) edge shared 
by their respective (primal) faces. The \emph{weight} of a primal edge gives rise to the \emph{length} (of same value) of its dual edge.
The cyclic order of the (dual) edges around 
any common incident (dual) vertex~$v^*$, is induced by the cyclic order of the 
(primal) edges around the (primal) face corresponding to~$v^*$.

We refer to a path/walk in $G_0^*$ as to a \emph{dual path/walk} in $G_0$,
and we speak about a {\em dual path/walk $\pi$ in $G_0$ between vertices $u,v$} if the
$\pi$ starts in a face incident with $u$ and ends in a face incident with~$v$.
We shortly say a {\em route from $u$ to $v$} (a {\em$u$--$v$ route})
to mean a dual walk between vertices $u,v$.

\begin{accumulate}
For any drawing $D$, let $\crdr D(X,Y)$ denote the number of crossings 
between edges of $X$ and edges of~$Y$ in $D$, and let $\crdr D(X):=\crdr D(X,X)$.
%
It is well established that the search for an 
optimal solution to the crossing number problem can be restricted to
so-called \emph{good} drawings: any pair of edges crosses at most once, adjacent 
edges do not cross, and there is no point that is a crossing of 
three or more edges. 
A simple extension of this finding to the setting of MEI is presented next,
in Lemma~\ref{lem:c+kch2}.

The following technical results will be used to restrict 
how ``complicated'' drawings of the edges of $F$ may look in an optimal
solution of a $\MEI(G,F)$ or $\rMEI(G,F)$ instance.
Note that, although both the claims are formulated for the rigid version,
they easily imply the same for the ordinary (non-rigid) MEI problem.
\begin{lem}\label{lem:c+kch2}
Consider a (weighted) instance $\rMEI(G,F)$ of the rigid MEI problem.
In any optimal solution of $\rMEI(G,F)$, any two edges of $F$ cross
at most once, and they have no crossing if they share a common endvertex.
Moreover, if the weights of the edges in $F$ equal~$1$ and 
there exists a drawing $D$ of $G+F$ such that
$\crdr{D}(E(G),F)=c$, then $\rINS(G,F)\leq c+{|F|\choose2}$.
\end{lem}
\begin{proof}
The proof simply repeats, for this special case of rigid MEI,
the folklore ``arc exchange'' argument from the crossing number theory.
For the second claim, we observe that since the edge weights in $F$ are
all~$1$, it is $\crdr{D}(F,F)\leq{|F|\choose2}$.
\end{proof}

\begin{cor}\label{cor:decreasef}
Consider a (weighted) instance $\rMEI(G,F)$ such that the weights of the
edges in $F$ equal~$1$, and let $f\in F$.
Assume that $D_1$ and $D_2$ are two drawings of $G+F$ such that
$D_1-f$ is identical to $D_2-f$,
and that $\crdr{D_1}(E(G),\{f\})\,-\,\crdr{D_2}(E(G),\{f\})>{|F|\choose2}$.
Then $\crd(D_1)>\rINS(G,F)$,
i.e., $D_1$ is not an optimal solution of $\rMEI(G,F)$.
\end{cor}
This claim might look rather weak at first sight,
with respect to the required large difference
$d:=\crdr{D_1}(f,E(G))-\crdr{D_2}(f,E(G))$.
However, one can actually easily construct examples in which $d=\Omega(|F|)$ and yet
$D_1$ is an optimal solution to $\rMEI(G,F)$.

\begin{proof}
Let $E=E(G)$, $k=|F|$ and $F'=F\setminus\{f\}$.
Using Lemma~\ref{lem:c+kch2}, we estimate
\begin{eqnarray*}
\rINS(G,F)&\leq& \crdr{D_2}(E,F) +{k\choose2} =
	\crdr{D_2}(E,F')+\crdr{D_2}(E,\{f\}) +{k\choose2}
\\ &=& \left[\, \crdr{D_1}(E,F')+\crdr{D_1}(E,\{f\}) \,\right]
	  -\crdr{D_1}(E,\{f\})+\crdr{D_2}(E,\{f\})+{k\choose2}
\\ &<& \crdr{D_1}(E,F)+0 \leq \crd(D_1)
\,.
\end{eqnarray*}
\end{proof}
\end{accumulate}

\section{Rigid MEI}\label{sec:rigidMEI}

In this section we give an FPT algorithm for solving the rigid version
$\rMEI(G,F)$, parameterized by $k=|F|$.
$G$ is hence a plane graph (i.e., with a fixed embedding) throughout this section.
Recall that the $\rMEI(G,F)$ problem is NP-hard~\cite{zieglerTh} for unrestricted~$k$.
%

We first illustrate the simple cases.
Solving $\rMEI(G,\{uv\})$, the fixed embedding edge insertion problem
with $k=1$, is trivial.
Augment dual $G^*$ with edges of length 0 between the terminals $u,v$ (technically, new vertices in $G^*$) 
and their respective incident faces (vertices in $G^*$),
to suit the above definition of a $u$--$v$ route in~$G$. 
\emph{Realizing} a route for $uv$ means to draw $uv$ along it within $G$.
If the shortest route has length $\ell$, 
realizing it attains $\rINS(G,\{vw\})=\ell$, the smallest number of crossings in the rigid MEI setting.


For $k\geq2$, the situation starts to be more interesting:
not every pair of shortest routes gives rise to an optimal solution
of $\rMEI(G,F)$ since there might arise a crossing between the two edges of~$F$.
The question, for $k=2$, is whether some pair of shortest routes 
of the two edges in $F$ can avoid crossing each other.
Since it is generally not feasible to enumerate all shortest routes, we
cannot check this by brute force and a more clever approach is needed.
Even worse, for larger values of $k$ we can encounter situations in which
optimal solutions of $\rMEI(G,F)$ draw edges of $F$
quite far from their individual shortest routes (in order to avoid
crossings with other edges of~$F$).

On a very high level, our approach to finding a drawing $D$ of $G+F$
that is an optimal solution to $\rMEI(G,F)$, can be described as follows:
\begin{enumerate}[(I)]\parskip0pt
\item \label{it:sch1}
We guess, for each pair $f,f'\in F$, whether $f$ and $f'$ will cross each other in $D$.
Since $k=|F|$ is a parameter, all the possibilities can be enumerated in FPT time.
\item 
Let $X\subseteq{F\choose2}$ be a (guessed) set of pairs of edges of~$F$.
We find a collection of shortest routes for the edges of $F$ in $G$ under the
restriction that exactly the pairs in $X$ cross;
$D$ is obtained by inserting the edges of $F$ along their computed routes.
As we will see, we may restrict our attention only to routes pairwise crossing at most
once.
\item \label{it:sch3}
We select $D$ which minimizes the sum of $|X|$ and of the lengths of the routes 
found above.
\end{enumerate}

\subsection{Handling path homotopy of routes}

The core task of the scheme \ref{it:sch1}--\ref{it:sch3} is to find a collection of shortest routes
under the restriction that every route avoids crossing certain other
routes (note; none of these routes are fixed in advance).
Although this problem may seem equivalent, in the dual, to the notoriously
hard problem of shortest disjoint paths in planar graphs
\cite{DBLP:journals/talg/VerdiereS11,DBLP:journals/disopt/KobayashiS10},
this is fortunately not the case since our routes may freely share their sections
as long as they do not cross.
We give a solution of the core task 
which is greedy in the sense
that each route of $F$ in~$G$ is minimized regardless of the other routes of~$F$.
The key to this solution is the concept of a {\em path homotopy} in the plane.

In a brief and rather informal topological view,
consider the sphere with a finite set of point obstacles.
Two simple curves $\alpha,\alpha'$ with the same endpoints are {\em homotopic}
if there exists a homeomorhpism (a continuous deformation) of $\alpha$ to
$\alpha'$ that fixes the endpoints and otherwise avoids all the obstacles.
For example, if $\alpha,\alpha'$ are disjoint except at the common ends, then
they are homotopic if and only if one of the two open regions bounded by
$\alpha\cup\alpha'$ is obstacle-free.
In our case, the {\em obstacles} are the ends $V(F)$ of the edges of $F$ 
(as given by the fixed embedding of~$G$), 
where each endpoint is ``blown up'' into a small open disc.
Then, given the homotopy classes $hom(\alpha),hom(\beta)$ of two curves
$\alpha,\beta$, one can decide whether $\alpha$ and $\beta$ are ``forced to
cross''---although, $\alpha$ and $\beta$ may cross if they are not forced
to, such an unforced crossing can as well be avoided in our case.

Instead of the above classical algebraic-topology setting of homotopies, 
in this paper we prefer to deal with path homotopy in a combinatorial setting.
This setting is closely inspired by the discrete-geometry view of
boundary-triangulated $2$-manifolds by Hershberger and 
Snoeyink~\cite{DBLP:journals/comgeo/HershbergerS94}.
In the first step, we ``triangulate'' the point set $V(F)$ (our obstacles)
using transversing paths in the embedding $G$.
A {\em transversing path} between vertices $x,y$ of $G$ is a path whose ends
are~$x,y$ and whose internal vertices subdivide some edges of~$G$.
Let $T$ be the union of these transversing paths and $G'$ denote the
corresponding subdivision of~$G$.
In order to avoid a terminology clash with graph triangulations,
we will call $T$ in the pair $(G',T)$ a {\em trinet of~$G$}.
Formally (where $V(F)=N$):

\begin{defn}[Trinet]\label{def:trinet}
Let $G$ be a connected plane graph and $N\subseteq V(G)$, $|N|\geq4$.
A plane graph $T$ such that $V(T)\cap V(G)=N$
is called a {\em trinet} of $G$ if the following holds:
\begin{enumerate}[a)]\parskip0pt
\item $T$ is a subdivision of a $3$-connected plane triangulation on the vertex set~$N$
(in particular, every face of $T$ is incident with precisely three vertices
of~$N$), and
\item there exists a subdivision $G'$ of $G$ such that
$V(G')\setminus V(G)=V(T)\setminus N$,~ $E(G')\cap E(T)=\emptyset$
and the union $G'\cup T$ is a plane embedding.
\end{enumerate}
The pair $(G',T)$ is a {\em full trinet} of~$G$.
The vertices in $N(T):=N$ are called {\em trinodes} of $T$, the maximal paths in $T$
internally disjoint from $N$ are {\em triedges} and their set is denoted by
$I(T)$, and the faces of $T$ are {\em tricells}.
Note that the triedges of $T$ are transversing paths of~$G$.
\ifnoaccumulating{We refer to Figure~\ref{fig:trinet-sketch} for a brief
illustration of this definition.}
\defend\end{defn}

Second, we focus on terms related to path homotopy in a full trinet $(G',T)$
of a plane graph~$G$.
Moreover, while we have implicitly perceived a route of $uv$ in $G$
(i.e., a dual walk from $u$ to $v$) as an arc drawn from $u$ to $v$,
we would also like to describe a topological ``alley'' for all
$u$--$v$ arcs of a similar kind (and same number of crossings)
in the embedding~$G'\cup T$.
With it we gain combinatorial abstraction and will later be able
to avoid unforced crossings with other routes.

\begin{defn}[Alley and $T$-sequence]\label{def:Tseq}
Let $(G',T)$ be a full trinet of a plane graph~$G$.
Consider a route $\pi$ between $u,v\in V(G)$ in the graph $G'\cup T$.
Then $V(\pi)=\{\phi_0,\phi_1,\dots,\phi_m\}$ where each dual vertex $\phi_i$
of $\pi$ is an open face of $G'\cup T$.
Let these faces $(\phi_0,\phi_1,\dots,\phi_m)$ be ordered along $\pi$ such that
$\phi_0$ is incident to $u$ and $\phi_m$ incident to $v$.
Let $(e_1,e_2,\dots,e_m)\subseteq E(G'\cup T)$ be the sequence of 
the primal edges of the dual edges of $\pi$, ordered from $\phi_0$ to $\phi_m$.
As a point set, each edge $e_i$ is considered without the endpoints.
\begin{enumerate}[a)]\parskip0pt
\item The union $\{u,v\}\cup\bigcup_{i=0}^m\phi_i\cup\bigcup_{i=1}^m e_i$
is called the {\em alley of $\pi$} (or, an {\em alley between $u,v$}).
\item 
Let $(e_1',\dots,e_\ell')\subseteq(e_1,e_2,\dots,e_m)$ be the restriction to $E(T)$,
and let $(p_1,p_2,\dots,p_\ell)\subseteq I(T)$ be the sequence of triedges
such that $p_i$ contains the edge $e_i'$ for~$i=1,\dots,\ell$.
Then $(p_1,p_2,\dots,p_\ell)$ is called the {\em $T$-sequence of $\pi$} from $u$ to $v$
(or, of the corresponding alley from $u$ to $v$).
\defend\end{enumerate}
\end{defn}

A route $\pi$ {\em crosses} a triedge $p$ if the alley of $\pi$ contains one of
the $G'$-edges forming~$p$.
The $T$-sequence of $\pi$ hence describes the unique order
(with repetition) in which its alley crosses the triedges of~$T$.
Usually, we shall consider only the case of $u,v\in N(T)$.


A route may, in general, cross the same triedge $q$ many times in one place
(switching ``there and back'').
However, such a situation may be easily smoothed down to one or no
crossing, and this can be formalized by the notion of reducing a $T$-sequence
as follows:
if $S=(p_1,p_2,\dots,p_\ell)$ is a $T$-sequence such that $p_i=p_{i+1}$
for some $1\leq i<\ell$, then the subsequence 
$S'=(p_1,\dots,p_{i-1},p_{i+2},\dots,p_\ell)$ is called a one-step reduction of $S$.
A subsequence $S^*\subseteq S$ is a {\em reduction} of $S$
(or {\em$S$ reduces to $S^*$})
if $S^*$ results from a sequence of one-step reductions of~$S$.
%

It comes as no surprise that $T$-sequences are closely related to the homotopy
concept:
\begin{rem}
Consider a trinet $T$ in the sphere.
One can show that two arcs with the same fixed endpoints are path-homotopic
(in the sphere with the obstacles formed by the trinodes of $T$)
if, and only if, their $T$-sequences can be reduced to the same subsequence.
However, since we are not going to directly use this fact,
we refrain from giving this as a formal statement in the short paper.
\end{rem}

\begin{accumulate}
\ifaccumulating{%
\subsection{Supplements for Section~\ref{sec:rigidMEI}}
}
\begin{figure}[p]
$$\includegraphics[width=0.65\hsize]{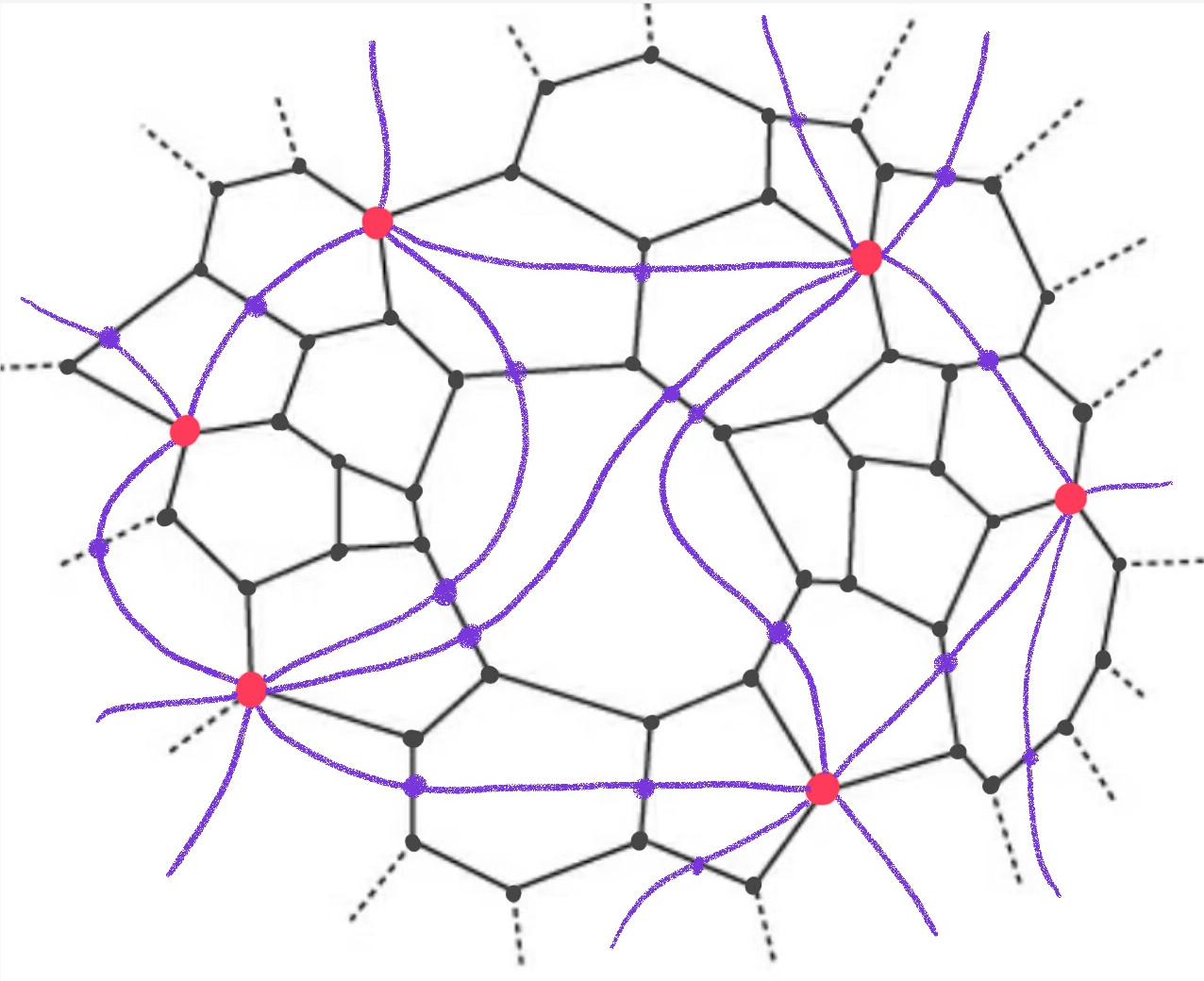}$$
\caption{An example of a trinet in a plane graph (see Definition~\ref{def:trinet}):
	the underlying plane graph $G$ is in black, the trinodes in thick
	red and the triedges in blue.}
\label{fig:trinet-sketch}
\end{figure}
\begin{figure}[p]
$$\includegraphics[width=0.65\hsize]{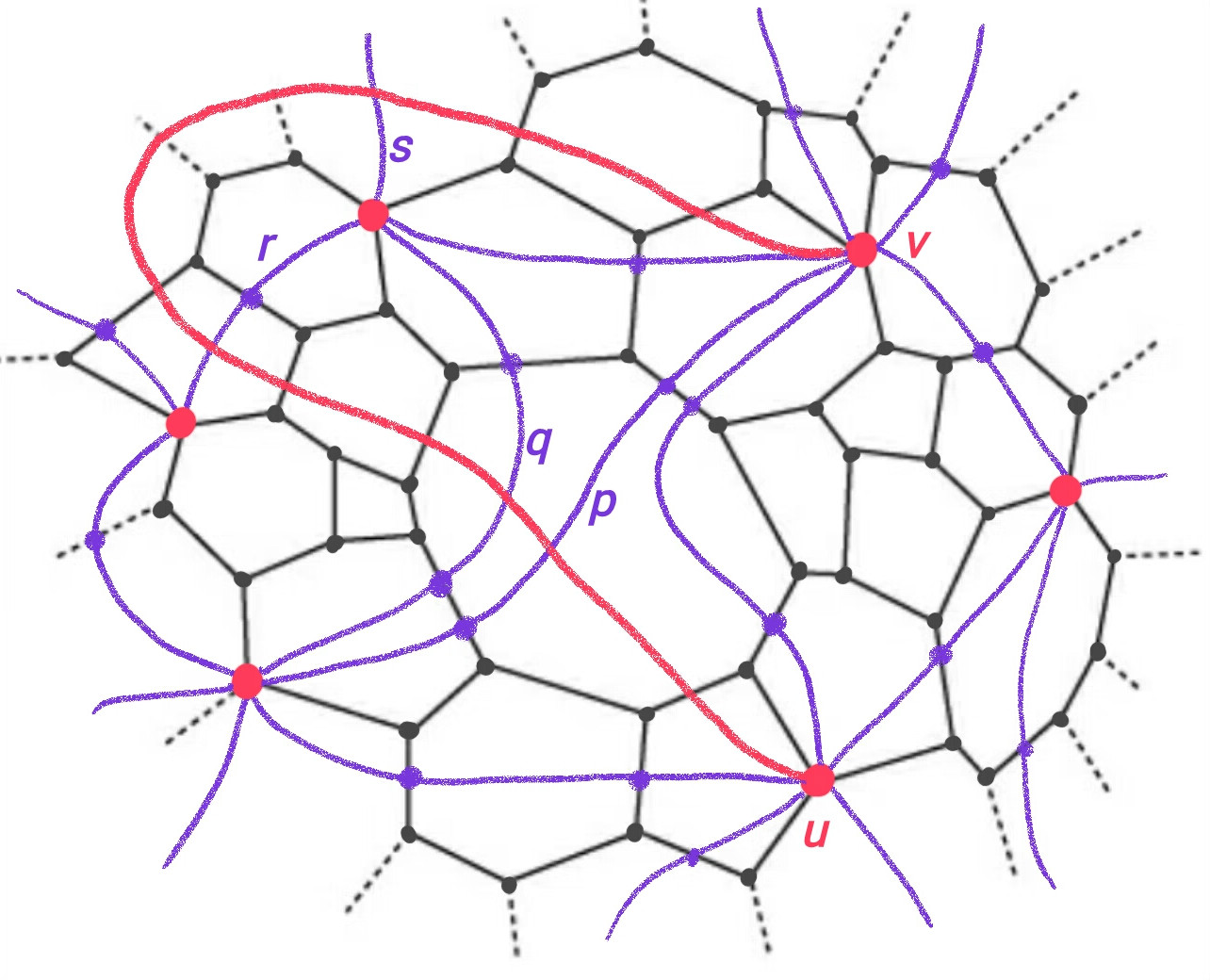}$$
\caption{An example (see Definition~\ref{def:Tseq}):
	the $T$-sequence of the $u$--$v$ route depicted in red
	is $(p,q,r,s)$ from $u$ to~$v$.
	It is a proper $T$-sequence from $u$ to~$v$
	(Definition~\ref{def:properTseq}).}
\end{figure}
\end{accumulate}


\subsection{$T$-sequences of potential shortest routes}

Our goal is to look for shortest routes in $G$ of a given homotopy, and
we slightly generalize the setting to allow for a connected plane graph $G$
with {\em edge weights} $w\colon E(G)\to\mathbb{N}_+\cup\{\infty\}$.
For a full trinet $(G',T)$ of~$G$,
we define the edge weights of $G'\cup T$ as follows:
$w'(p):=0$ for all $p\in E(T)$ and $w'(e'):=w(e)$
where $e'\in E(G')$ is obtained by subdividing~$e\in E(G)$.
This $w'$ is the {\em weight induced by $w$} in the trinet $(G',T)$.
We give the same weights $w'$ also to the edges of the geometric dual of
$G'\cup T$.
If $\alpha$ is the alley of a route $\pi$ between vertices $x,y$ in $G'\cup T$,
then the {\em length of $\alpha$} equals the length of $\pi$,
i.e., the sum of the $w'$-weights of the dual edges of~$\pi$.
%

With the help of the framework developed in the previous section,
we can now give an (again informal) high-level refinement of our solution steps
\ref{it:sch1}--\ref{it:sch3} of $\rMEI(G,F)$ as follows:
\begin{enumerate}[(I)]\parskip0pt
\setcounter{enumi}3
\item\label{it:onealley}
Consider a trinet $T$ of $G$ on the trinodes $V(F)$.
If we fix a (realizable) $T$-sequence $S$,
then we can use established tools, namely an adaptation of the idea of the funnel
algorithm~\cite{DBLP:conf/focs/Chazelle82,DBLP:journals/networks/LeeP84},
to efficiently compute a shortest alley among those 
having the same $T$-sequence~$S$.
For $uv\in F$ of weight~$1$, if we compute an alley $\alpha$ between $u,v$ of length
$\ell$, then we can easily draw the new edge $uv$ as an arc in $\alpha$ 
with $\ell$ weighted crossings.
%
\item\label{it:whichcross}
Suppose that, for $i=1,2$, $\alpha_i$ is a shortest alley between $x_i$ 
and $y_i$ having the $T$-sequence~$S_i$.
Then, as detailed later in Lemma~\ref{lem:noncrossing} and also
Claim~\ref{clm:indepcertif},
we can decide from {\em only $S_1,S_2$} whether there exist arcs from $x_1$ to
$y_1$ in $\alpha_1$ and from $x_2$ to $y_2$ in $\alpha_2$, which do not cross
(note that $\alpha_1\cap\alpha_2$ may be nonempty and yet there may exist
such a pair of non-crossing arcs).
Moreover, if the two arcs cross then it should be only once.
\item\label{it:boundTseq}
Consequently, it will be enough to loop through all ``suitable'' $T$-sequences
for every edge of $F$ and independently perform the steps 
\ref{it:onealley}, \ref{it:whichcross} for each combination of them, 
in order to get an optimal solution of $\rMEI(G,F)$ as in~\ref{it:sch3}.
The point is to bound the number of $T$-sequences that have to be considered,
in terms of only the parameter~$k=|F|$.
\end{enumerate} 

We first resolve the last point \ref{it:boundTseq} which is a purely mathematical
question.
In order to achieve the goal,
we will build a special trinet of $G$ along (at least locally) shortest dual paths 
between the trinodes in $G$ (Definition~\ref{def:locally-shortest}).
Then we will be able to restrict our attention to special $T$-sequences
of bounded length (Definition~\ref{def:properTseq} and Lemma~\ref{lem:shortTseq}).

\begin{defn}[Shortest-spanning trinet]\label{def:locally-shortest}
Let $(G',T)$ be a full trinet of a plane graph~$G$,
and let the weights $w'$ in $(G',T)$ be induced by weights $w$ in $G$.
For a triedge $q\in I(T)$, every internal vertex $t$ of $q$ is incident
with two edges $e,e'$ of $G'$ of weight $w'(e)=w'(e')$
which we call the weight of~$t$.
The {\em transversing weight of $q$} equals the sum of
the weights of the internal vertices of~$q$.

A triedge $q\in I(T)$ between trinodes $x,y$ is {\em locally-shortest}
if the transversing weight of $q$ is equal to the length 
of a shortest dual path $\pi$ in $G'\cup T$ between $x,y$,
such that $\pi$ is contained in(!) the union of the two tricells incident with~$q$.
Similarly, $q$ is {\em globally-shortest}
if the transversing weight of $q$ is equal to the dual distance between
$x,y$ in $G'\cup T$.

We say that $T$ has the {\em shortest-spanning property}
if every triedge in $I(T)$ is locally-shortest,
and there exists a subset of triedges $J\subseteq I(T)$ forming a connected
subgraph of $T$ spanning all the trinodes 
such that every triedge in $J$ is globally-shortest.
\defend\end{defn}


\begin{defn}[Proper $T$-sequence]\label{def:properTseq}
Consider a trinet $T$ and trinodes $u\not=v\in N(T)$.
A {\em nonempty} sequence $S=(p_1,p_2,\dots,p_m)\subseteq I(T)$ of triedges of $T$
(repetition allowed) is a {\em proper $T$-sequence from $u$ to $v$}
if the following holds:
$u$ is disjoint from $p_1$ but there exists a tricell $\theta_0$
incident with both $u$ and $p_1$,
$v$ is disjoint from $p_m$ but there exists a tricell $\theta_m$
incident with both $v$ and $p_m$, and each two consecutive 
triedges $p_i,p_{i+1}$ are distinct and incident to a common tricell
$\theta_i$ for $1\leq i<m$.
{\em Empty} $S$ is a {\em proper $T$-sequence from $u$ to $v$}
if $u,v$ are incident to a common tricell $\theta_0$.
\defend\end{defn}

Recalling that $T$ is a subdivision of a triangulated graph,
we immediately get the following:
\begin{claim}\label{cl:Tseq-tricells}
For every proper nonempty $T$-sequence $S$,
the sequence of tricells $(\theta_0,\theta_1,\dots,\theta_m)$
as in Definition~\ref{def:properTseq} is uniquely determined by $T$ and~$S$.
\qed\end{claim}\noindent
For {\em empty} proper $S$, a tricell $\theta_0$ from
Definition~\ref{def:properTseq} is not unique (there are two choices for
it). However, since any of the two choices of $\theta_0$ incident with both
$u,v$ will work for us in the same way in a shortest-spanning trinet,
we simply make an arbitrary deterministic choice of $\theta_0$ for
$S=\emptyset$ and extend the scope of Claim~\ref{cl:Tseq-tricells} 
also to empty proper $T$-sequences.

What follows is the crucial finding that makes our algorithm to work:

\begin{lem}\apxmark\label{lem:shortTseq}
Consider an instance $\rMEI(G,F)$ where $G$ is a connected
plane graph. 
Let $(G',T)$ be a full trinet of~$G$ having the short\-est-spanning property.
There exists a set $\{\pi_f: f\in F\}$ where $\pi_f$ for $f=uv$ is a
route in $G'\cup T$ between the trinodes $u,v$,
such that the following hold:
\begin{enumerate}[a)]\parskip0pt
\item
There exists an optimal drawing $D$ of $G+F$ with $\rINS(G,F)$ crossings
such that each edge $f\in F$ is drawn in the alley of~$\pi_f$,
and no two edges of $F$ cross each other more than once.
\item
The $T$-sequence $S_f$ of each $\pi_f$ is a proper $T$-sequence,
and no triedge occurs in $S_f$ more than $8k^4$ times where~$k=|F|$.
\end{enumerate}
\end{lem}

\begin{accumulate}
\ifaccumulating{\subsubsection*{Proof of Lemma~\ref{lem:shortTseq}}}
Note that optimality of the number $\rINS(G,F)$ ensures that, in a), no two edges of
$F$ cross in $D$ unless they are forced to (by their given alleys).

\begin{proof}
In the scope of this proof, we shall use the following special terminology
and notation.
For simplicity, we use the symbol $f$ both for an edge $f\in F$ and
for the arc representing $f$ in a specific drawing of $G+F$
(more generally, of $(G'+F)\cup T$).
We similarly consider a triedge $p\in I(T)$ also as the arc representing $p$ in~$G'$.
If $x,y$ are two points on any arc $b$, then let
$b[x,y]$ denote the section of the arc from $x$ to~$y$.

Consider the given shortest-spanning trinet and the corresponding 
plane embedded graph $G'\cup T$ with edge weights $w'$ induced by given $w$ of~$G$.
We will implicitly assume that every arc $a$ drawn in $G'\cup T$ avoids crossing
the vertices and $a$ intersects $G'\cup T$ in finitely many points, 
i.p., the embedding and the arcs may be restricted to polygonal lines.
For any arc $b$ with the ends $u,v$, 
we define the {\em $T$-sequence of $b$} from $u$ to $v$
as the sequence (with repetition) in which $b$ intersects the triedges of~$T$.
We define the {\em transversing weight of $b$}, shortly {\em t-weight}, 
as the sum of the $w'$-weights of the edges of $G'\cup T$ crossed by~$b$,
and denote it by $t_{w'}(b)$.

We choose an optimal drawing $D$ of $G'+F$ which, at the same time,
minimizes the combined length of the $T$-sequences of the edges of~$F$,
i.e., the number of crossings between $F$ and the trinet~$T$.
Recall from Lemma~\ref{lem:c+kch2}: any two edges of $F$ cross
at most once, and they have no crossing if they share a common endvertex.
For $f\in F$, let $S_f=(p_f^1,\dots,p_f^{m_f})$ be the $T$-sequence of~$f$
and let $x_f^1,\dots,x_f^{m_f}$, respectively, 
denote the points at which the arc of $f$ intersects the triedges of $T$.
The first task is to prove that each $S_f$ is a proper $T$-sequence.

We start with a stronger technical claim:
if, for some $f\in F$ and $j>i$, it is $p_f^i=p_f^j$ and the simple loop
$a:=p_f^i[x_f^i,x_f^j]\cup f[x_f^i,x_f^j]$ is contractible
(i.e., with no trinode inside), then we get a contradiction to the choice
of~$D$ above.
Indeed, we may assume that $f$ and $j>i$ are chosen such that $a$ encloses
minimal area in the drawing~$D$.
By the minimality of $a$, no triedge crosses the interior of $f[x_f^i,x_f^j]$
twice (all $p_f^{i+1},\dots,p_f^{j-1}$ are distinct).
However, since the interior enclosed by $a$ contains no trinode,
the previous implies that no triedge other than $p_f^i$ may intersect $a$,
and so $j=i+1$.
Consequently, since the triedge $p_f^i$ is locally shortest in $T$,
the t-weights of the considered section satisfy 
$t_{w'}\big(p_f^i[x_f^i,x_f^j]\big) \leq t_{w'}\big(f[x_f^i,x_f^j]\big)$.
If we re-route $f$ closely along $p_f^i[x_f^i,x_f^j]$ (without crossing $p_f^i$),
then this change does not increase the crossing number
by the inequality of t-weights, but the $T$-sequence of $f$ gets shorter
(see in Figure~\ref{fig:shorten-ex1}).
Hence, it contradicts our choice of $D$.

Now we get back to $S_f$ being a proper $T$-sequence.
If $S_f$ is empty, then the statement is trivial.
If $S_f$ contains consecutive repeated triedge $p_f^i=p_f^{i+1}$ for
some $1\leq i<m_f$, then the above contradiction directly applies.
Assume now that $f=uv$ and the triedge $p_f^1$ is incident with the
starting trinode~$u$.
Then we can apply the same contradiction to the contractible loop
$a:=p_f^1[u,x_f^1]\cup f[u,x_f^1]$.
The remaining properties of proper $T$-sequences follow trivially.

\medskip
The last and most difficult step is to prove that no triedge repeats
in $S_f$ too many times, for each $f\in F$.
Again, if $p_{f}^i=p_{f}^j=p'$ for $i\not=j$, then the simple loop
$a':=p'[x_{f}^i,x_{f}^j]\cup f[x_{f}^i,x_{f}^j]$ must be non-contractible,
and so separating some pair of trinodes of $T$ from each other.
In particular, by assumed connectivity of~$G'$, 
this implies that always $t_{w'}(a')\geq1$.
Since at most $2k-1$ globally-shortest triedges of $T$ span all the trinodes
by Definition~\ref{def:locally-shortest}, 
$a'$ (and consequently $f[x_{f}^i,x_{f}^j]$) must cross at least one of them.
Therefore, if a triedge $p'$ repeats in $S_{f}$ at least $8k^4$ times,
then there is a globally-shortest triedge $p$ of $T$ such that $p$ repeats
in $S_f$ at least $\ell\geq8k^4/(2k-1)>4k^3$ times.

Let $Y=(y_1,\dots,y_\ell)$ (a subsequence of $(x_f^1,\dots,x_f^{m_f})$)
be the ordered sequence of points in which the arc of $f$ intersects the arc
of the triedge~$p$.
We say that an index $i\in\{2,\dots,\ell-1\}$ is a {\em switchback} of $Y$
if $y_{i-1},y_{i+1}$ both lie on the same side of $y_i$ on~$p$.
Up to symmetry, let the points on $p$ be ordered such that
$y_{i+1}$ lies between $y_{i-1},y_i$.
Since $p$ is globally-shortest in $T$,
we get (now regardless of contractibility of the induced loops)
$$
t_{w'}\big(p[y_{i-1},y_i]\big) \leq t_{w'}\big(f[y_{i-1},y_i]\big)
,$$
and then
\begin{eqnarray*}
t_{w'}\big(f[y_{i-1},y_{i+1}]\big) &\geq&
	t_{w'}\big(p[y_{i-1},y_i]\big)+t_{w'}\big(f[y_i,y_{i+1}]\big)
\\	&=&
  t_{w'}\big(p[y_{i-1},y_{i+1}]\big)+ t_{w'}\big(p[y_{i+1},y_i]\big)
	+ t_{w'}\big(f[y_i,y_{i+1}]\big)
\\	&\geq&
  t_{w'}\big(p[y_{i-1},y_{i+1}]\big) +1
.\end{eqnarray*}

Hence, if we locally re-route $f$ along $p[y_{i-1},y_{i+1}]$,
then we save the amount of at least $1$ in the crossings of~$f$ with $E(G)$.
Note that this is not a contradiction to our choice of optimal drawing $D$
yet since the change may introduce many new crossings of $f$ with the rest
of~$F$.
However, we cannot have more than $k\choose2$ switchbacks in $Y$ or we get a
contradiction using Corollary~\ref{cor:decreasef}.
(Observe that this usage of the corollary here is the only reason why we 
are restricting ourselves to unweighted sets $F$.)

Since $\ell\geq4k^3$, there is a consecutive subsequence $Y'\subseteq Y$
of length $\ell'\geq\ell/{k\choose2}-1>8k$ without switchbacks.
Without loss of generality, we assume $Y'=(y_1,\dots,y_{\ell'})$.
Let $g_i:=f[y_i,y_{i+1}]$ and 
$g_i^\circ:=g_i\cup p[y_{i},y_{i+1}]$, for $i\in\{1,\dots,\ell'-1\}$.
As argued before, each $g_i^\circ$ is a simple loop separating some 
pair of trinodes of $T$.
Since no two edges of $F$ cross more than once,
there are at most $k-1$ indices $i\in\{1,\dots,\ell'-1\}$ such that
$g_i$ is crossed by another edge(s) of~$F$.

Let $x,y$ be the ends of the triedge~$p$.
Assume that we have $i\not=j\in\{1,\dots,\ell'-1\}$ such that
neither of $g_i^\circ,g_j^\circ$ separates $x$ from~$y$.
Let $Z_i\not=\emptyset$ denote the set of trinodes of $T$ that 
are separated by $g_i^\circ$ from $x,y$, and let $Z_j$ be defined analogously.
We claim that $X_i\cap X_j=\emptyset$.
If not, then---up to symmetry---$g_j^\circ$ is separated from $x,y$ by
$g_i^\circ$, except a possibly shared section of $p[y_{i},y_{i+1}]$.
The former is impossible by the Jordan curve theorem and the
latter would mean that there is a switchback between $i$ and $j$,
which is again a contradiction.
Since there are at most $2k-2$ pairwise disjoint nonempty possibilities
(e.g., singleton trinodes other than $x,y$) for the sets $X_i,X_j$,
at most $2k-2$ indices $i\in\{1,\dots,\ell'-1\}$ are such that
$g_i^\circ$ does not separate $x$ from~$y$.

Since $\ell'\geq8k$, there exists a set of indices
$J\subseteq\{1,\dots,\ell'-2\}$, $|J|\geq \ell'-2(k-1+2k-2)-2>2k$,
such that for every $j\in J$ both the arcs $g_j,g_{j+1}$
are not crossed by other edges of~$F$ and both $g_j^\circ,g_{j+1}^\circ$
separate $x$ from~$y$.
Let $f_0:=f[y_{1},y_{\ell'}]$ and $p_0:=p[y_{1},y_{\ell'}]$;
we get $Y'\subseteq p_0$ since there are no switchbacks in~$Y'$.
Observe also that $g_j^\circ\cap g_{j+1}^\circ=\{y_{j+1}\}$
since $f$ is not self-intersecting and there is no switchback in~$Y'$.
Hence, up to symmetry, $g_j^\circ$ separates $x$ from $g_{j+1}^\circ$,
and $g_{j+1}^\circ$ separates $g_j^\circ$ from~$y$.
It easily follows that $g_j^\circ\cup g_{j+1}^\circ$ forms
the boundary of an arc-connected region of 
$\mathbb R^2\setminus(f_0\cup p_0)$ (a {\em face of} $f_0\cup p_0$).
Since at most $2k-2$ of the faces of $f_0\cup p_0$ may contain a trinode of
$T$ other than $x,y$, there exists $j\in J$
such that, in addition to the above properties of $J$,
the face $\sigma$ bounded by $g_j^\circ\cup g_{j+1}^\circ$ contains no trinode
(see in Figure~\ref{fig:shorten-ex2}).

Our goal now is to re-route $f$ along $p[y_{j},y_{j+1}]$ (i.e.,
``replacing'' the part $g_j\subset f$).
Again, since $p$ is globally-shortest in $T$, this move does not increase
the number of crossings of $f$ with $E(G)$, and the $T$-sequence of $f$
gets shorter.
It remains to argue that we can avoid new crossings of $f$ with
$F\setminus\{f\}$.
If any $f'\in F$ crosses $p[y_{j},y_{j+1}]$ then, since $\sigma$ contains no
trinode, $f'$ has to leave $\sigma$ as well, and the only possibility
is across $p[y_{j+1},y_{j+2}]$ by the previous assumptions.
Consequently, such $f'$ can be re-routed along 
$p[y_{j},y_{j+2}]$, similarly to~$f$, and no crossing with $f$ is required
(see again in Figure~\ref{fig:shorten-ex2}).
Note, moreover, that even if two such edges $f',f''\in F$ cross each other
in $\sigma$, there is no problem and they will cross in their new routing in
the same way.
We have again reached a contradiction to our choice of~$D$.
\end{proof}

\begin{figure}[p]
$$\includegraphics[width=0.66\hsize]{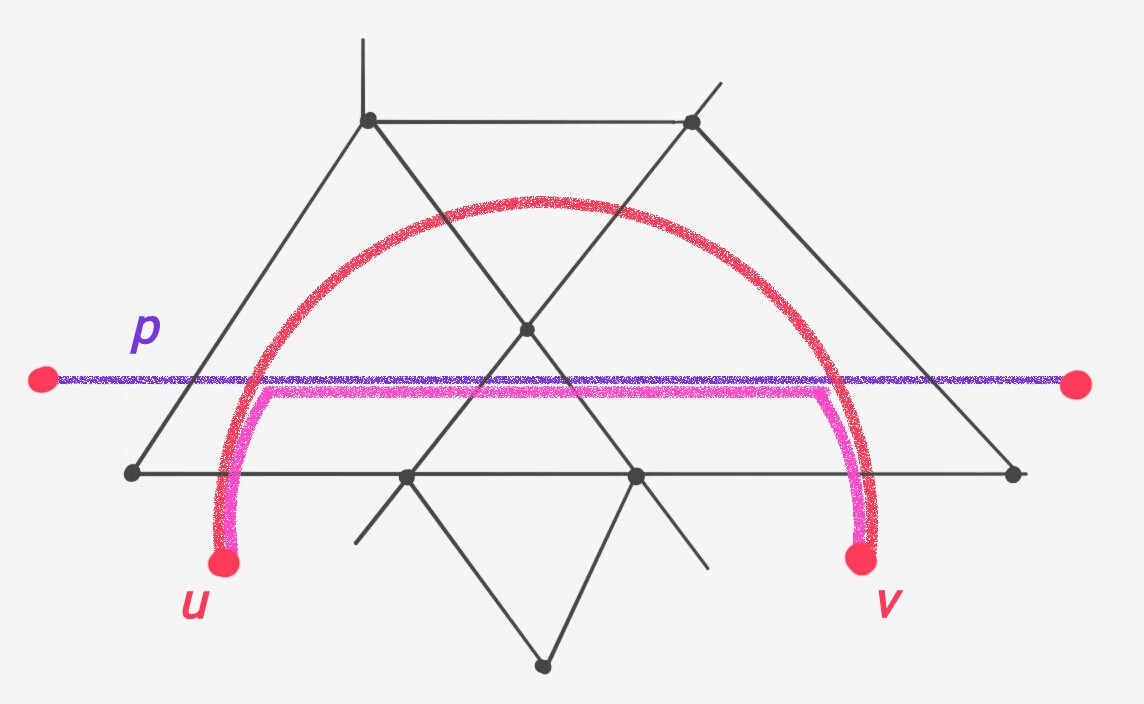}$$
\caption{Two consecutive crossings of the arc of $f=uv$ with a triedge~$p$
	(which is locally-shortest)
	determine a contractible loop, and so $f$ can be re-routed partly
	along $p$ without inducing more crossings with $G$ or with other
	edges of~$F$.}
\label{fig:shorten-ex1}
\end{figure}
\begin{figure}[p]
$$\includegraphics[width=0.66\hsize]{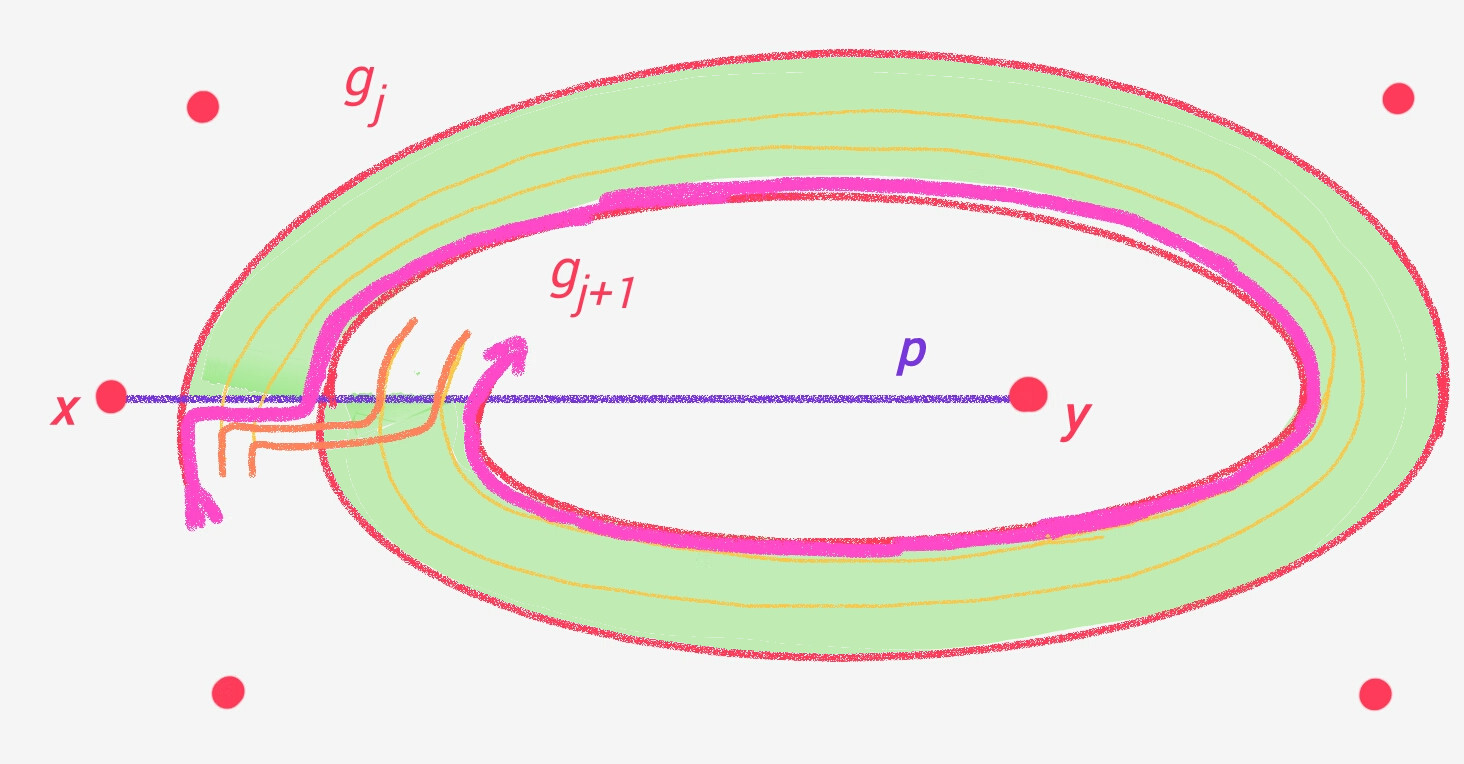}$$
\caption{The face bounded by $g_j^\circ\cup g_{j+1}^\circ$
	(\,$\subseteq f\cup p$) depicted in green:
	since there is no trinode in this face and neither $g_j,g_{j+1}$
	are crossed by other edges of $F$, it is possible to re-route
	$f$ partly along $p$ such that it avoids $g_j$.
	Other possible $F$-edges entering the green face through a section
	of $p$ must leave at the other end, and hence can be re-routed
	similarly to $f$.}
\label{fig:shorten-ex2}
\end{figure}
\end{accumulate}

\subsection{Shortest routes in a sleeve, and crossing of routes}

Now, consider step \ref{it:onealley} of our outline---%
as mentioned before, we solve this step separately for each $f=uv\in F$.
To recapitulate, for trinodes $u,v$ of a trinet $T$ of~$G$
and a given proper \mbox{$T$-sequence} $S$ from $u$ to $v$,
the task is to find a shortest route from $u$ to $v$
among those having the same $T$-sequence~$S$.
We cannot, in general, completely avoid repeating triedges in $S$ and tricells in the
sequence $(\theta_0,\theta_1,\dots,\theta_m)$ in Definition~\ref{def:properTseq}.
To prevent related technical difficulties,
we use a similar workaround as in~\cite{DBLP:journals/comgeo/HershbergerS94};
``lifting'' the respective sequence of tricells into a universal cover as
follows.

\begin{defn}[Sleeve of a $T$-sequence]\label{def:sleeve}
Let $(G',T)$ be a full trinet of a plane graph~$G$,
and consider a proper $T$-sequence $S=(p_1,p_2,\dots,p_m)$ from $u$ to $v$
determining the sequence of tricells
$(\theta_0,\theta_1,\dots,\theta_m)$ by Claim~\ref{cl:Tseq-tricells}.
For $i=0,1,\dots,m$, let $L_i$ be a disjoint copy of the embedded subgraph
of $G'\cup T$ induced by~$\theta_i$.
Construct a plane graph $L$ from the union $L_0\cup\dots\cup L_m$
by identifying, for $j=1,\dots,m$, the copy of the triedge $p_j$ in
$L_{j-1}$ with the copy of $p_j$ in $L_j$.
We call $L$ the {\em sleeve of $S$} in the trinet $(G',T)$,
and we identify $u$ and $v$ with their copies in~$L_0$ and $L_m$, respectively.
We make the unique face of $L$ that is not covered by a copy of any tricell
of~$T$ the \emph{outer face} of~$L$.
\defend\end{defn}

Observe that every route from $u$ to $v$ in $G'\cup T$
having its $T$-sequence equal to $S$ can be easily 
lifted into a corresponding $u$--$v$ route in the sleeve~$L$ of~$S$.
Conversely, any $u$--$v$ route in $L$ avoiding the outer face and
crossing the copies of triedges in $L$ at most once each, 
can be obviously projected down to $G'\cup T$ to make a route with the
$T$-sequence equal to $S$.
In fact, we prove 
that some shortest
$u$--$v$ route in $L$ must be of the latter kind, under the shortest-spanning
property (cf.~Definition~\ref{def:locally-shortest}).

\begin{lem}\apxmark\label{lem:pathinsleeve}
Let $(G',T)$ be a shortest-spanning full trinet of an edge-weighted plane graph~$G$,
$S$~a~proper $T$-sequence between trinodes $u,v$ of $T$, 
and let $L$ be the sleeve of~$S$.
Let $\ell$ be the length of a shortest route from $u$ to $v$
among those having the $T$-sequence~$S$. Then, $\ell$ is equal to the dual distance from
$u$ to $v$ in $L$ without the outer face.
Furthermore, at least one of the $u$--$v$ routes of length~$\ell$ in $L$ crosses
the copy of each triedge from $S$ in $L$ exactly once.
\end{lem}
\begin{accumulate}
\ifaccumulating{\subsubsection*{Proof of Lemma~\ref{lem:pathinsleeve}}}
\begin{proof}
Let $\pi$ be any shortest route from $u$ to $v$ in $G'\cup T$
with the given $T$-sequence~$S$.
The copies of the faces and dual edges of $\pi$ lifted into the sleeve $L$
give a route $\pi_L$ from $u$ to $v$ which avoids the outer face of~$L$.
Obviously, the length of $\pi_L$ equals the length of~$\pi$.

Conversely, we aim to show that some shortest $u$--$v$ route crosses
the copy of each triedge of $S$ in $L$ exactly once.
Assume a shortest route $\pi_1$ of length $\ell$ 
from $u$ to $v$ in $L$ without the outer face.
Recall from Definition~\ref{def:sleeve} that $L=L_0\cup\dots\cup L_m$.
For $S=(p_1,p_2,\dots,p_m)$, let $(p_1',p_2',\dots,p_m')$ be the sequence of
corresponding copies of the triedges of $S$ in~$L$,
and let $p_0':=u$,~$p_{m+1}':=v$.
Note that each $p_i'$, $i\in\{1,\dots,m\}$, connects two vertices of the
outer face of~$L$, and so $p_i'$ separates $p_{i-1}'$ from $p_{i+1}'$.
In particular, every $u$--$v$ route in $L$ which avoids the outer face
must cross each of $p_1',\dots,p_m'$.

Let $i$ be the maximum index such that $p_i'$ is crossed by $\pi_1$ more
than once.
Then there is a subpath $\sigma_1\subseteq\pi_1$
stretching between two consecutive crossings of $\pi_1$ with $p_i'$
and contained in~$L_i$.
We turn $\pi_1$ into $\pi_2$ by re-routing the subpath $\sigma_1$
along the boundary $p_i'$ in~$L_{i-1}$.
Since $p_i$ is locally-shortest in the trinet $T$,
the length of $\pi_2$ equals the length of $\pi_1$.
By induction on the number of excess crossings of $\pi_2$ with copies of
the triedges, we can then get a $u$--$v$ route $\pi_0$ 
of length~$\ell$ such that $\pi_0$ crosses
the copy of each triedge from $S$ in $L$ exactly once.

Finally, the route $\pi_0$ projects down to a route
of length $\ell$ from $u$ to $v$ in $G'\cup T$ having the $T$-sequence~$S$.
\end{proof}
\end{accumulate}

Regarding the shortest path computation,
we note that if the edge-weights are given in unary or are bounded by a constant, 
a simple adaption of BFS achieves the job. 
Otherwise, we can use the algorithm of Klein et al.~\cite{KleinEA} 
since $L$ is planar, or Thorup's algorithms~\cite{Thorup} since we have integral weights.
Altogether we obtain: 

\begin{cor}\apxmark\label{cor:pathinsleeve}
Let $(G',T)$ be a shortest-spanning full trinet of an edge-weighted plane graph~$G$,
and $S$~a~proper $T$-sequence between trinodes $u,v$ of $T$.
A shortest $u$--$v$ route among those having the $T$-sequence~$S$
can be found in the geometric dual graph of the sleeve $L$ of $S$ 
in $\OO(|S|\cdot|N(T)|\cdot|V(G)|)$ time, using a linear time shortest path algorithm.
\end{cor}
\noindent
Observe that in our case the term $|S|\cdot|N(T)|$ is bounded by a function of~$k=|F|$.
\begin{accumulate}
\ifaccumulating{\subsubsection*{Proof of Corollary~\ref{cor:pathinsleeve}}}
\begin{proof}
By Definition~\ref{def:trinet},
the size of the full trinet $(G',T)$ is $\OO(|N(T)|\cdot|V(G)|)$,
and the sleeve is composed of $|S|$ copies of subgraphs of $G'\cup T$,
and so $\OO(|S|\cdot|N(T)|\cdot|V(G)|)$ bounds the size of the sleeve~$L$.
Therefore, any linear time shortest path algorithm applicable in our
situation, e.g. the aforementioned Klein et al.~\cite{KleinEA} or 
Thorup~\cite{Thorup}, does the job.
\end{proof}
\end{accumulate}

\smallskip
Finally, it remains to address step \ref{it:whichcross}.
Consider a $4$-tuple of distinct trinodes $u,v,u',v'$.
%
Let $\pi$ be a $u$--$v$ route and $\pi'$ be a $u'$--$v'$ route.
We say that an {\em arc $b$ follows the route $\pi$}
if $b$ is contained in the alley of $\pi$ and $b$ intersects the faces
forming the alley exactly in the order given by $\pi$
(recall that a route is technically a dual walk and hence, possibly, 
some face might repeat in~$\pi$).
We say that the pair of routes $\pi,\pi'$ is {\em non-crossing},
if there exist a $u$--$v$ arc $b$ following $\pi$
and a $u'$--$v'$ arc $b'$ following $\pi'$
such that $b\cap b'=\emptyset$.
In order to characterize possible non-crossing pairs of routes
in terms of their $T$-sequences, we bring the following definition:

\begin{defn}[Crossing certificate]\label{def:crosscert}
Let $(G',T)$ be a full trinet of a plane graph~$G$,
and let $\pi$ be a route from $u$ to $v$ and $\pi'$ be a route
from $u'$ to $v'$ in $G'\cup T$,
where $u,v,u',v'$ are distinct trinodes of~$T$.
Assume the $T$-sequences $S=(p_1,\dots,p_n)$ of $\pi$ and
$S'=(p_1',\dots,p_\ell')$ of $\pi'$ are proper and let
$(\theta_0,\dots,\theta_n)$ and
$(\theta_0',\dots,\theta_\ell')$ be their tricell 
sequences by Claim~\ref{cl:Tseq-tricells}.
For technical reasons, let $p_0:=u$, $p_{n+1}:=v$
and $p_0':=u'$, $p_{\ell+1}':=v'$.

A {\em crossing certificate for $S,S'$} is a triple of indices
$(c,d,m)$ where $c,d,m\geq0$, $c+m\leq n$, $d+m\leq\ell$,
such that the following holds:
\begin{enumerate}[a)]\parskip0pt
\item $\theta_{c+j}=\theta'_{d+j}$ for $0\leq j\leq m$,
but $p_c\not=p'_d$ and $p_{c+m+1}\not=p'_{d+m+1}$,
\item the triple $p_c,p_d',p_{c+1}$ occurs around the tricell
$\theta_c=\theta'_d$ in the same cyclic orientation as
the triple $p_{c+m+1},p'_{d+m+1},p_{c+m}$ occurs around
$\theta_{c+m}=\theta'_{d+m}$.
\end{enumerate}
Furthermore, a crossing certificate for the same sequence $S$ and the reversal of
$S'$ from $v'$ to $u'$ is also called a crossing certificate for $S,S'$.
\defend\end{defn}

Definition~\ref{def:crosscert} deserves a closer explanation.
Assume that a crossing certificate satisfies $0<c<n$ and $0<d<\ell$.
Then all four elements $p_c,p'_d,p_{c+1},p'_{d+1}$ are triedges of the same
tricell $\theta_c=\theta'_d$, and since
$p_{c+1}\not=p_c\not=p'_d\not=p'_{d+1}$, we get $p_{c+1}=p'_{d+1}$.
Hence $m>0$ and the situation is such that $S$ and $S'$ ``merge'' at
$\theta_c$ where (up to symmetry) $S$ comes on the left of $S'$,
and they again ``split'' at $\theta_{c+m}$ where $S$ leaves on the right
of~$S'$, thereby ``crossing it''.
The full definition, though, covers also the boundary cases of crossing
certificates for which $c\in\{0,n\}$ or $d\in\{0,\ell\}$ (or both),
and when $S$ and $S'$ may have no triedge in common;
those can be easily examined case by case.


\begin{lem}\apxmark\label{lem:noncrossing}
Let $(G',T)$ be a full trinet of an edge-weighted plane graph~$G$,
and $u_i,v_i$, $i=1,2$, be four distinct trinodes.
Assume that $S_i$ from $u_i$ to $v_i$ are proper $T$-sequences.
In $G'\cup T$, for $i=1,2$, there exist routes $\pi_i$ from $u_i$ to $v_i$
having the the $T$-sequence~$S_i$, 
such that {\em$\pi_1,\pi_2$ are non-crossing},
if and only if there exists no crossing certificate for~$S_1,S_2$.
\end{lem}

\begin{accumulate}
\ifaccumulating{\subsubsection*{Proof of Lemma~\ref{lem:noncrossing}}}
Suppose we have two proper $T$-sequences $S,S'$ as in Definition~\ref{def:crosscert}.
Referring to Definition~\ref{def:crosscert}, point~a),
we call the tricells $\theta_{c+j}=\theta'_{d+j}$ for $0\leq j\leq m$
the {\em central tricells of the crossing certificate} $(c,d,m)$.

\begin{proof}
Let $L_i$, $i=1,2$, be the sleeves of $S_i$ in $(G',T)$.
Assume that $(c,d,m)$ is a crossing certificate for $S_1,S_2$,
and let $R=(\theta_{c},\dots,\theta_{c+m})$ be the sequence of the central
tricells of this certificate.
Let $K_i\subseteq L_i$, $i=1,2$, be the plane subgraphs consisting of the
copies of the tricells from $R$ in the sleeve $L_i$.
Note that $R$ may repeat the same tricell several times, but in $L_i$
we have got independent copies of the possibly repeated tricells.
We may also assume that $K_1=K_2$ since they are both made of copies of the
same sequence $R$ of tricells.

In the above view, Definition~\ref{def:crosscert} says that $(c,d,m)$ is a
crossing certificate iff the elements $p_c,p_d',p_{c+m+1},p'_{d+m+1}$
appear on the outer face of $K_1=K_2$ in this cyclic order.
Hence, by Jordan's curve theorem,
if there is a crossing certificate for~$S_1,S_2$, then $\pi_1,\pi_2$
cannot be non-crossing.

\smallskip
Conversely, we show how to build non-crossing $\pi_1,\pi_2$ if there is no
crossing certificate for $S_1,S_2$.
For each tricell $\theta$ of $T$, bounded by triedges $q_1,q_2,q_3$,
we choose arbitrary three edges $e_j$ from $q_j$, $j=1,2,3$,
and arbitrary three internally disjoint dual paths
$\sigma_{1,2},\sigma_{1,3},\sigma_{2,3}$ contained in $\theta$
such that $\sigma_{i,j}$ is a dual path in $G'\cup T$ 
connecting the face incident with $e_i$ to the face incident with~$e_j$.
Furthermore, we denote by $x_j$ the trinode of $\theta$ opposite to $q_j$
and we choose another three arbitrary dual paths $\rho_1,\rho_2,\rho_3$
contained in $\theta$ such that $\rho_j$ is a dual path in $G'\cup T$ 
connecting a face incident with $x_j$ to the face incident with~$e_j$.
We call chosen $\sigma_{1,2},\sigma_{1,3},\sigma_{2,3},\rho_1,\rho_2,\rho_3$
the representative dual paths of the tricell~$\theta$.

For the proper $T$-sequence $S_i$, $i=1,2$,
we simply compose the route $\pi_i$ of the apropriate representative dual paths 
of the tricells determined by $S_i$.
It is routine to verify that these $\pi_1,\pi_2$ are non-crossing,
if and only if there exists no crossing certificate for~$S_1,S_2$.
\end{proof}
\end{accumulate}

\subsection{Summary of the algorithm}

We are now ready to put all of the above results together, in order to summarize the overall
algorithm to solve r-MEI, see Algorithm~\ref{alg:rmei}.
Based thereon, together with Lemmas~\ref{lem:shortTseq},\,\ref{lem:pathinsleeve},
Corollary~\ref{cor:pathinsleeve}, and Lemma~\ref{lem:noncrossing} we obtain: 
\begin{thm}\apxmark\label{thm:rigidalg}
Let $G$ be a connected plane graph with edge weights $w\colon E(G)\to\mathbb{N}_+\cup\{\infty\}$,
and $F$ a set of new edges (vertex pairs, in fact) such that $w(f)=1$ for all~$f\in F$.
Algorithm~\ref{alg:rmei} finds an optimal solution to
$w$-weighted $\rMEI (G,F)$, if a finite solution
exists, in time $\OO\big(2^{\poly(|F|)}\cdot|V(G)|\big)$.
\end{thm}

\begin{pseudoalg}[h]
{\centering\begin{minipage}{0.99\hsize}\small
\begin{itemize}\item[]
\begin{enumerate}[(1)]
\item[\bf Input:] a plane graph $G$, edge weights
 $w\colon E(G)\to\mathbb{N}_+\cup\{\infty\}$,
new edge set $F$ s.t.~$w(f)=1$ for~$f\in F$.
\item[\bf Output:] an optimal solution to ($w$-weighted) $\rMEI(G,F)$.
\item Compute a full trinet $(G',T)$, $N(T):=V(F)$,
	with the shortest-spanning property of~$T$.
\label{it:ctrinet}
\begin{enumerate}[(a)]
\item Pick any trinode $x\in N(T)$ and greedily compute globally-shortest
	triedges (Def.~\ref{def:locally-shortest}) from $x$ to all other
	trinodes, using a simple shortest path computation.
\item The remaining triedges can be  greedily computed as locally-shortest,
	 one after another.
\end{enumerate}
\item\label{it:scriS} For each $f=uv\in F$:
\begin{enumerate}[(a)]
\item Compute $\scri S_f$ as the set of all of its possible proper $T$-sequences
	from $u$ to~$v$.
	The size of $\scri S_f$ is bounded due to Lemma~\ref{lem:shortTseq}(b)
	as $\OO(2^{\poly(|F|)})$.
\item For each $S\in \scri S_f$, compute a shortest $u$--$v$ route $\pi_S$
	in $G'\cup T$ among those having the $T$-sequence $S$
	(where the length function is induced by~$w$), using Corollary~\ref{cor:pathinsleeve}.
\end{enumerate}
\item For each possible set $\scri P = \{S_f\}_{f\in F}$ with $S_f\in\scri S_f$:
\label{it:scriP}
\begin{enumerate}[(a)]
\item\label{it:crossingcert}
	Check, for each pair $f,f'\in F$, whether there
	exists a crossing certificate for $S_f,S_{f'}$
	(\mbox{e.g., using} brute force by Def.~\ref{def:crosscert}).
	\\Let $X_{\scri P}$ be the set of pairs $\{f,f'\}$ for which such a certificate
	has been found.
\item\label{it:crossingtwice}
	If any pair $\{f,f'\}\in X_{\scri P}$ requires more than
	a single crossing (which can be found by checking again for
	two ``independent'' \ifaccumulating{(~~\apxmark) }%
	crossing certificates of $S_f,S_{f'}$),
	let $\crg_{\scri P}:=\infty$.
\item\label{it:crossingsum}
	Otherwise, let
	 $\crg_{\scri P}:= |X_{\scri P}|
		 + \sum_{f\in F}\, len_{w}(\pi_{S_f})$,
	where $len_{w}$ denotes the length function in the geometric dual of
	$G'\cup T$ induced by~$w$.
\end{enumerate}
\item\label{it:choosemin} Among all $\scri P$ considered in~\ref{it:scriP},
	pick the one with smallest $\crg_{\scri P}<\infty$.
	Let this be $\scri P^\circ=\{S^\circ_f\}_{f\in F}$.
\item\label{it:realize}
	In the plane graph $G$, realize each edge $f\in F$ following its respective
	route $\pi_{S^\circ_f}$, such that the overall resulting weighted number 
	of crossings is $\crg_{\scri P^\circ}$.
\begin{enumerate}[(a)]
\item\label{it:postp}
	(By minimality, no $\pi_{S^\circ_f}$ will be self-intersecting.)
	Using well-known postprocessing---removing consecutive
	crossings between $f,f'$ by re-routing $f'$ partially along $f$ or vice
	versa---allows to avoid multiple crossings in pairs from $X_{\scri P}$
	and to make remaining pairs from $F$ crossing-free.
\end{enumerate}
\end{enumerate}
\end{itemize}
\end{minipage}\par\medskip}
\caption{Algorithm to solve the (weighted) rigid MEI problem.}
\label{alg:rmei}
\end{pseudoalg}

\begin{accumulate}
\ifaccumulating{\subsubsection*{Proof of Theorem~\ref{thm:rigidalg}}}
\noindent
Before giving the proof,
we need a deeper understanding of the concept of non-crossing routes and
crossing certificates,
and a detailed specification of the step \ref{it:scriP}\ref{it:crossingtwice}
of Algorithm~\ref{alg:rmei}.

By adapting the arguments of Lemma~\ref{lem:noncrossing}, one can actually
get the following slight strengthening:
\begin{claim}\label{clm:crossingr}
Let $(G',T)$ be a full trinet of an edge-weighted plane graph~$G$,
and $u_i,v_i$, $i=1,2$, be four distinct trinodes.
Let $\pi_i$, $i=1,2$, be a $u_i$--$v_i$ route in $G'\cup T$.
If there exist simple $u_i$--$v_i$ arcs $b_i$, $i=1,2$, 
following the route $\pi_i$, such that
$b_1$ intersects $b_2$ in exactly one point $x$ and they properly cross in~$x$,
then there exists a crossing certificate for the $T$-sequences of $\pi_1$
and~$\pi_2$.
\qed\end{claim}

We say that there exist {\em two independent crossing certificates}
for the $T$-sequences $S,S'$ if there are crossing certificates
$(c,d,m)$ and $(c',d',m')$ for $S,S'$
(each one up to possible reversal of $S'$ as in Definition~\ref{def:crosscert}),
such that the set of central tricells of $(c,d,m)$ is disjoint from
the set of central tricells of $(c',d',m')$.
The following can then be straightforwardly obtained from
Lemma~\ref{lem:noncrossing}:
\begin{claim}\label{clm:indepcertif}
Let $(G',T)$ be a full trinet of an edge-weighted plane graph~$G$,
and $u_i,v_i$, $i=1,2$, be four distinct trinodes.
Assume that $S_i$ from $u_i$ to $v_i$ are proper $T$-sequences.
In $G'\cup T$, there exist simple arcs $b_i$ from $u_i$ to $v_i$
such that, for $i=1,2$,
\begin{itemize}\parskip0pt
\item[--]$b_i$ is contained in the alley of a $u_i$--$v_i$ route
having the $T$-sequence~$S_i$, and
\item[--]$b_1$ intersects $b_2$ in at most one point,
\end{itemize}
if and only if there exists no two independent crossing certificates for~$S_1,S_2$.
\qed\end{claim}

The implementation of step \ref{it:scriP}\ref{it:crossingtwice},
using Claim~\ref{clm:indepcertif}, hence
simply checks by brute force for the existence of two independent crossing
certificates for~$S_f,S_{f'}$.

\begin{proof}[Proof of Theorem~\ref{thm:rigidalg}]
In this proof, we will use some of the terminology and notation from the proof 
of Lemma~\ref{lem:shortTseq}, and refer to the notation
of Algorithm~\ref{alg:rmei}.

Consider arbitrary $\scri P = \{S_f\}_{f\in F}$ as in the step~\ref{it:scriP}.
The value of $\crg_{\scri P}$ computed in the step
\ref{it:scriP}\ref{it:crossingsum}, provided that $\crg_{\scri P^\circ}<\infty$,
is a lower bound on the number of crossings of any feasible solution of
$\rMEI(G,F)$ such that, for each $f\in F$, the $T$-sequence of the arc
$f$ is exactly~$S_f$.
This fact follows directly from Lemma~\ref{lem:noncrossing}
(for the part $|X_{\scri P}|$)
and from Lemma~\ref{lem:pathinsleeve}
\smallskip
(for the part $\sum_{f\in F} len_{w}(\pi_{S_f})$).

By Lemma~\ref{lem:shortTseq}, there is an optimal feasible solution $D$ to
$\rMEI(G,F)$ such that, for every $f\in F$,
the arc of $f$ in $D$ has its $T$-sequence 
(with respect to the trinet $T$ from step \ref{it:ctrinet})
equal to some proper $S_f\in\scri S_f$ as computed in \ref{it:scriS},
and the step \ref{it:scriP}\ref{it:crossingtwice} does not apply to these
values by Claim~\ref{clm:indepcertif}.
Consequently, $\crg_{\scri P^\circ}\leq\rINS(G,F)$ by the lower-bound
argument from the previous paragraph.
Hence if we can prove that the step \ref{it:realize} indeed can compute
a drawing $D^\circ$ of $G+F$ with $\crg_{\scri P^\circ}$ weighted crossings,
provided $\crg_{\scri P^\circ}<\infty$,
then we complete the proof of Theorem~\ref{thm:rigidalg}.

\smallskip
For $f\in F$, let $b_f$ denote a realization of the edge $f$
as an arc following the route of $\pi_{S^\circ_f}$
(such that the $T$-sequence of $b_f$ is $S^\circ_f$),
before the postprocessing step \ref{it:realize}\ref{it:postp}.
By the minimality choice in step~\ref{it:choosemin}, we can be sure that
$b_f$ does not cross itself: the self-crossing would induce a non-contractible loop with at least one crossing over $G'$,
but then there exists a $T$-sequence $S'_f\subset S^\circ_f$---the sequence $S^\circ_f$ without the triedges forcing the loop.
Replacing these two sequences in $\scri P^\circ$ results in a smaller
number of crossings on $f$ while not increasing the number of crossings at any term in the summation considered in step~\ref{it:scriP}\ref{it:crossingsum}.

Let $D_1$ denote the drawing of $G+F$ made of $G$ and $\{b_f:f\in F\}$.
Observe that $\crdr{D_1}(E(G),F)=
	\sum_{f\in F}\, len_{w}(\pi_{S_f^\circ})$.

Fix some~$f,f'\in F$, $f=uv$ and $f'=u'v'$, 
such that $b_f,b_{f'}$ cross each other (properly) more than once.
Consider the point set $p:=b_f\cup b_{f'}$
with the outer face incident with, say, the trinode~$u$.
If any of the bounded faces of $p$ contained a trinode (which cannot be~$u$)
of~$T$, then we could ``split'' the $T$-sequences $S^\circ_f$ and $S^\circ_{f'}$
into $S^1_f,S^2_f$ and $S^1_{f'},S^2_{f'}$ each, such that
for each of the pairs $S^1_f,S^1_{f'}$ and $S^2_f,S^2_{f'}$
there would exist a crossing certificate by Claim~\ref{clm:crossingr}.
This would, in turn, provide two independent crossing certificates for
the $T$-sequences $S^\circ_f,S^\circ_{f'}\in{\scri P^\circ}$ and
by the step \ref{it:scriP}\ref{it:crossingtwice},
it would contradict $\crg_{\scri P^\circ}<\infty$.

Consequently, all the bounded faces of $p=b_f\cup b_{f'}$ are free of
trinodes of~$T$.
This, in particular, means that if we construct $b_f^1$ from $b_f$ by
re-routing it along a section of $b_{f'}$ between two consecutive shared
points of $b_f\cap b_{f'}$, then the $T$-sequence of $b_f^1$ would
again be $S^\circ_f$.
Moreover, since $b_f,b_{f'}$ have been chosen from their respective
shortest routes, the t-weight of $b_f^1$ would be equal to
$len_{w}(\pi_{S^\circ_f})$ (which is the t-weight of original~$b_f$).
Iterating this process, we arrive at a drawing $D_2$ of $G+F$ satisfying the
following:
\begin{enumerate}[i.]\parskip0pt
\item no two edges of $F$ in $D_2$ cross more than once,
\item $\crdr{D_2}(E(G),F)\leq\crdr{D_1}(E(G),F)$.
\end{enumerate}
What remains is to observe that two edges $f,f'\in F$ properly cross each other in $D_2$
only if $\{f,f'\}\in X_{\scri P^\circ}$.
Indeed, if $f,f'$ properly cross in $D_2$, then this crossing is the only one
and there exists a crossing certificate for $\pi_{S^\circ_f},\pi_{S^\circ_{f'}}$
by Claim~\ref{clm:crossingr}, and then $\{f,f'\}\in X_{\scri P^\circ}$
due to the step \ref{it:scriP}\ref{it:crossingcert}.

To summarize, we get
\begin{eqnarray*}
\crg_{\scri P^\circ} &\leq&
	\rINS(G,F) \leq \crd(D_2)=\crdr{D_2}(E(G),F)+\crdr{D_2}(F,F)
\\	&\leq&
	\crdr{D_1}(E(G),F)+|X_{\scri P^\circ}|
	= \crg_{\scri P^\circ}
\end{eqnarray*}
which proves optimality of the solution computed by
Algorithm~\ref{alg:rmei}.

\medskip
Finally, we discuss the runtime bound of Algorithm~\ref{alg:rmei}.
Let $k=|F|$.
Step \ref{it:ctrinet} is performed in time $\OO(k^2\cdot|V(G)|)$ using
$3(2k)-6$ calls to a linear shortest path algorithm.
Step \ref{it:scriS} takes time
$\OO(k\cdot2^{poly(k)}\cdot k|V(G)|)$
by Corollary~\ref{cor:pathinsleeve}.
Step \ref{it:scriP} is iterated $\OO(2^{poly(k)\cdot k})$ times,
and each iteration takes time polynomial in $k$ (independently of $G$)
even by brute force.
Step \ref{it:choosemin} takes only time $\OO(2^{poly(k)\cdot k})$.
Finally, step \ref{it:realize} performs $k$ computations in $\OO(|V(G)|)$,
to realize each $f\in F$ in~$G$, and then
a number of concurrent re-routings which can be bounded by an amortized analysis:
every of the $k$ routes is of length $\OO(|V(G)|)$ and each element of it
could be re-routed at most once towards each of the $k-1$ remaining routes,
summing to $\OO(k^2\cdot|V(G)|)$.

The above analysis sums up to overall $\OO\big(2^{poly(k)}\cdot|V(G)|\big)$ time.
\end{proof}
\end{accumulate}

\section{General MEI}\label{sec:nonrigid}
\begin{accumulate}
\ifaccumulating{\subsection{Supplements for Section~\ref{sec:nonrigid}}}
\end{accumulate}

Now, we may turn our attention to the general MEI problem, where the embedding of the planar graph $G$ is not prespecified.
See also the appendix for details.
Recall that triconnected planar graphs have a unique embedding (up to mirroring), but already biconnected graphs
have an exponential number of embeddings in general.
As it is commonly done in insertion problem since \cite{GMW05}, we will use the \emph{SPR-tree} datastructure (sometimes also known as SPQR-tree) to encode and work with all 
these possible embeddings. It was first defined in slightly different form in~\cite{DT96}, based on prior work of~\cite{BM90,Tut66}. 
It can be constructed in linear time~\cite{HT73,GM01} and only requires linear space.

\begin{defn}[SPR-tree, cf.~\cite{mchDiss}]\label{def:SPR}
Let $G$ be a biconnected graph with at least three vertices. 
The \emph{SPR-tree} $\scri T$ of $G$ is the unique smallest 
tree satisfying the following properties: 
\begin{enumerate}[i)]
\item Each node $\nu$ in $\scri T$ holds a specific (small) graph $S_\nu=(V_\nu,E_\nu)$,
with $V_\nu\subseteq V(G)$, called a \emph{skeleton}.
Each edge $e$ of $E_\nu$ is either a {\em real} edge $e\in E(G)$,
or a {\em virtual} edge $e=xy\not\in E(G)$ (while still, $x,y\in V(G)$).
  \item $\scri T$ has three different node types with the following skeleton structures:
{\bf (S)} $S_\nu$ is a simple cycle;
{\bf (P)} $S_\nu$ consists of two vertices and
at least three multiple edges between them;
{(\bf R)} $S_\nu$ is a simple triconnected graph
on at least four vertices.
\item\label{it:emuenu}
For every edge $\nu\mu$ in $\scri T$ we have $|V_\nu\cap V_\mu|=2$. 
These two common vertices, say $x,y$, form a vertex $2$-cut
(a {\em split pair}) in~$G$. 
Skeleton $S_\nu$ contains a specific virtual edge $e_\mu\in E(S_\nu)$ that 
represents the node $\mu$
and, symmetrically, some specific $e_\nu\in E(S_\mu)$ represents~$\nu$;
both $e_\nu,e_\mu$ have the ends $x,y$. 
\item\label{it:gluev}
The original graph $G$ can be obtained by recursively applying the following 
operation of {merging}:
For an edge $\nu\mu\in E(\scri T)$, let $e_\mu$, $e_\nu$ be the 
pair of virtual
edges as in (\ref{it:emuenu} connecting the same $x,y$.
A {\em merged} graph \mbox{$(S_\nu\cup S_\mu)-\{e_\mu,e_\nu\}$}
is obtained by gluing the two skeletons together at $x,y$ and removing $e_\mu,e_\nu$.\defend
\end{enumerate}
\end{defn}

The central theorem of~\cite{GMW05} states that we can find an optimal embedding to insert a single edge $uv$ by
looking at the shortest path in $\scri T$ between a node whose skeleton contains $u$ and a node whose skeleton contains $v$.
For each skeleton along this path, one considers the partial routes between the virtual edge representing $u$ (or $u$ itself) and
the virtual edge representing $v$ (or $v$ itself). In case of S- and P-nodes this route requires no crossings (by choosing a suitable embedding 
in the latter case); for an R-node $\nu$, the route is a shortest path in the dual of its skeleton: if the primal edge is an original edge, 
the length of its dual edge is the primal edge's weight;
if the primal edge is a virtual edge $xy$, representing node $\mu$, the length of its dual edge is the minimum-$xy$-cut in $P_\mu$, where we 
$P_\mu$ is the \emph{pertinent graph of $\mu$} arising from merging all skeletons of the subtree rooted at~$\mu$, minus the edge $e_\nu$.
By picking \emph{any} embedding of $P_\mu$ and computing a shortest dual path through it, we can compute this cut size in linear time.
See~\cite{GMW05} for details. 

We consider our SPR-tree $\scri T$ of $G$ rooted at any node, and devise a dynamic programming scheme 
to solve MEI bottom-up over $\scri T$. We observe that 
every non-root skeleton $S_\nu$ contains a virtual edge $e_\varrho$ that represents its father in $\scri T$. Any further
virtual edges correspond to children of $\nu$ in $\scri T$.
Since we already know how to solve r-MEI, it shall suffice to describe
which r-MEI problems we need to solve at each SPR-tree node $\nu$ (including the root node), assuming we already solved the corresponding subproblems at their children.
The overall MEI solution can then be obtained by selecting a solution in the root with the least number of crossings. 

We say a virtual edge in $S_\nu$ is \emph{dirty} if
it contains an end vertex of a new edge $f\in F$.
Hence, at most $2k$ edges (a constant number) of $S_\nu$ are dirty.
For R- and S-nodes, we only have to consider their unique (up to mirroring, in case of R) embeddings. A P-node whose skeleton contains 
$p$ edges, however, allows $(p-1)!$ embeddings. Based on the following claim (which can be shown with a straight-forward redrawing argument), we only need to consider up to $(2k)!$ embeddings for each P-node,
which is constant for constant~$k$.
\begin{claim} 
Let $\nu$ be a P-node in the SPR-tree $\scri T$ of $G$. There is an optimal embedding of $G$ for the MEI problem, where all non-dirty virtual edges
are consecutive in the embedding of $S_\nu$.
\end{claim}

Let $S'_\nu$ be a considered embedding of $S_\nu$. Consider each virtual edge
$xy$, representing node $\mu$ (possibly, $\mu$ is the father node), in $S_\nu$.
If $xy$ is not dirty, we set its weight to the size of the minimum-$xy$-cut in
$P_\mu$ (as for the single edge insertion case). If $xy=e$ is dirty, we modify
it with the following gadget: Set the weight of $e$ to $\infty$, and add two new
\emph{side edges} $e',e''$ connecting $x$ and~$y$. One is directly to the left,
the other directly to the right of $e$. We will further modify these side edges
in the following.

Consider what can happen at the subdrawing of (embedded) component $P_\mu$ in
the context of whole $G$ when considering any specific new edge $f\in F$: (i) if
$f$ has exactly one end in $P_\mu$, it will enter the component; (ii) if $f$ has
no end in $P_\mu$, it may cross through the component; (iii) if $f$ has both
ends in $P_\mu$, it may leave the component and re-enter $P_\mu$ at another
position. Furthermore, and in contrast to the single edge insertion, it may
happen that $f$ (independent of its end points) crosses $P_\mu$ multiple times.
However, since we consider a fixed embedding $S'_\nu$, we know from Lemma~\ref{lem:shortTseq} that
the latter number is bounded by a constant, depending only on $k$. Hence, there
are only a bounded number of enterings/leavings at $P_\mu$, and we can simply
consider all such possible situations (including all possible orders of the
enterings/leavings). For each such situation, we now subdivide the edges $e'$
and $e''$ accordingly: chiefly put, if, e.g., we consider the case of an edge
$f=uv$ coming from a vertex $u\not\in V(P_\mu)$ and crossing $P_\mu$ twice before finally entering it to reach $v\in
V(P_\mu)$, we generate (for $f$) overall five vertices on $e',e''$, say $v_1,\ldots,v_5$. Within the
context of the dynamic programming subproblems at $P_\mu$, we then consider (for
each embedding of $S_\mu$) the r-MEI problem w.r.t.\ \emph{subedges}
$v_1v_2$, $v_3v_4$, $v_5v$. Overall, we have to store the best 
solution (over all embeddings of $S_\mu$) for each r-MEI
problem constructed of all such subedges, for each edge $f\in F$, each
possible number of crossings of $f$ through $P_\mu$, each possible assignment of
thereby induced subdivision vertices to $e'$ or $e''$, and all possible orders
at $e'$ and $e''$. Within a subproblem at $\nu$, we then consider the edge 
$uv_1$ instead of $f$ (in fact, this edge may be further split into several 
subedges due to further dirty virtual edges in $S_\nu$ considered to be crossed 
by $f$, and/or if $u$ is contained in a pertinent graph of another virtual 
edge).

Hence, we only need to store a constant (bounded by a function in $k$) number
of solutions at each SPR-tree node. Each solution can be obtained using the above algorithm for r-MEI
in $\OO(|V(G)|)$ time, and there are at most $\OO(|V(G)|)$ SPR-tree nodes. Instead of the na\"ive quadratic runtime bound,
we even achieve a linear runtime bound by observing that
the union of all skeletons is still only of linear size. 
We obtain, as given in the introduction:
\begin{thm}[The biconnected case of Theorem~\ref{thm:main}]\label{thm:bicon}
 Let $G$ be a planar biconnected graph on $n$ vertices, and $F$ a set of $k$ new edges (vertex 
pairs, in fact) where $k$ is a constant.
We can solve $\MEI(G,F)$ in $\OO(n)$ time. 
\end{thm}

For essentially all known insertion algorithms (in particular single edge insertion~\cite{GMW05}, vertex 
insertion~\cite{insertvertex}, and MEI approximation~\cite{CHicalp}), one can typically first describe
the case of biconnected graphs (using SPR-trees). Then, it is relatively straight-forward to 
lift the algorithms to connected graphs, by considering BC-trees (see below). Interestingly, this seems much
more complicated in case of exact MEI:

Consider the well-known block-cut tree (BC-tree) to decompose any connected graph into its blocks (biconnected components).
Using analogous techniques as in \cite{CHicalp}, we extend our dynamic programming approach by amalgating the BC-tree with
the blocks' respective SPR-trees, to obtain a linear-sized \emph{con-tree}, with an additional node type \textbf{C}, for cut vertices.
In our bottom-up approach, at a cut vertex $c$, we need to consider all possibilities to ``glue'' the $c$-incident dirty blocks (blocks with at least one end of some edge $f\in F$) 
together.  However, we cannot easily bound this number by a function purely in $k$: we not only have to consider all orders of these blocks, but also all possible nestings, which introduces a dependency on $\Delta_\mathrm{cut}$,
the maximum degree of the cut vertices in $G$. Hence, for only connected $G$, we obtain the slightly weaker result:
\begin{thm}[The connected case of Theorem~\ref{thm:main}]\label{thm:connected}
 Let $G$ be a planar connected graph on $n$ vertices, and $F$ a set of $k$ new edges (vertex 
pairs, in fact), where $k$ and the maximum degree of the cut vertices in $G$ are constant.
We can solve $\MEI(G,F)$ in $\OO(n)$ time. 
\end{thm}


\begin{accumulate}
As sketched in the main body of the paper, we first develop a dynamic programing algorithm over the SPR-tree decomposition 
$\mathcal{T}$ of planar $G$. This algorithm considers \emph{dirty nodes} bottom-up; a decomposition node $\nu$ is \emph{dirty}
if its pertinent graph contains at least one vertex incident to $F$. Observe that if a node is dirty, so is its parent. 
The root node (whose pertinent graph we may define as $G$ itself) is always dirty.

\paragraph{Subproblems at non-root nodes.}
We start with formally defining the subproblems to be solved and stored at each dirty non-root decomposition node.
Let $\nu$ be such a node and let $e=xy\in E(S_\nu)$ be the virtual edge in the skeleton of $\nu$ corresponding 
to its parent node $\varrho$. Recall that the pertinent graph $P_\nu$ arises from $S_\nu$ by merging the skeletons
of the subtree rooted at $\nu$ and removing the sole remaining virtual edge ($e$).
We consider the 3-partition of $F$ into $F_0,F_1,F_2$, 
where $F_0$ are the edges without an end in $V(P_\nu)\setminus\{x,y\}$,
$F_1$ are the edges with one end in $V(P_\nu)\setminus\{x,y\}$ and the other not in $V(P_\nu)$,
and $F_2$ are the edges with one end in $V(P_\nu)\setminus\{x,y\}$ and the other in $V(P_\nu)$.

By definition, the graph $P_+:=P_\nu+e$ is planar, and $e$ represents the ``rest of the graph'' disjoint from $P_\nu$.
We are, intuitively, interested in the best embedding $P^\circ_+$ of $P_+$ 
to 
\begin{enumerate}[(a)]
 \item route the edges of $F_1$ from a side of $e$ to its end in $V(P_\nu)\setminus\{x,y\}$; observe that we may care
from which side of $e$ the new edge emanates.
\end{enumerate}
 But these are not the only routes to
consider in an optimal solution: 
\begin{enumerate}[(a)]\addtocounter{enumi}{1}
 \item edges $uv\in F_2$ may be routed completely within $P_\nu$, or go from $u$ to some side of $e$ (into the ``rest of the graph''), and from some side $e$ (from the ``rest of the graph''; either the same or the other side) to $v$;
 \item any edge of $F$ may be routed through $P_\nu$, i.e., from one side of $e$ to the other side, without crossing $e$.
\end{enumerate}
Formally, we can define a \emph{routing query} as a pair $(s,t)$, where $s$ and $t$ are each either referencing
a specific side of $e$ or a vertex in $V(P_\nu)\setminus\{x,y\}$. We will use $e',e''$ to denote the two different sides of $e$.
In such a routing query, we ask for a routing of a new edge between $s$ and $t$ in $P_+$, without crossing over $e$.

Lemma~\ref{lem:shortTseq} (which holds for every fixed embedding, and hence for each possible embedding) showed that a triedge of trinet $T$ is crossed at most 
$8k^4=\OO(\poly(k))$ times. When computing a shortest route (w.r.t.\ some $T$-sequence) between two succeeding triedges, we clearly have the property
that any edge within the corresponding tricell is crossed at most once. Hence:
\begin{cor}
 In an optimal solution to $\MEI(G,F)$, each edge $f\in F$ crosses any edge $e\in G$ at most $\xi:=\poly(k)$ times.
\end{cor}
Since this corollary also holds for virtual edges in a skeleton, we have the same upper bound for crossings
through a two-connected component $P_\nu$.

In our dynamic programming scheme, we will hence---for each possible set of routing queries---store the minimum number of crossings necessary
over all embeddings of $P_+$. 
A specific set of routing queries (to be described in details below) is hence a \emph{subproblem}, and the corresponding number of crossings (together with 
the embedding of $P_+$ and the corresponding routings, if desired) is a \emph{subsolution}. It remains to discuss the number of subproblems for $\nu$. 

\begin{lem}
Each subproblem specifies at most $\OO(\poly(k))$ routing queries.
The total number of subproblems to consider at any node $\nu$ is bounded by $\OO(\poly(k)!)$.
\end{lem}
\begin{proof}
Consider the routing types (a)--(c) as above:
\begin{enumerate}[(a)]
\item For each edge $f=uv\in F_1$ with $u\in V(P_\nu)\setminus\{x,y\}$, we have to pick one of the two routing queries $(e',u)$, $(e'',u)$.

\item For each edge $f=uv\in F_2$, we have to pick one out of five options: (i) a single routing query $(u,v)$; (ii)--(v) two routing queries
$(u,e^{(1)})$, $(e^{(2)},v)$, with $e^{(1)},e^{(2)}\in\{e',e''\}$.

\item Finally, for each $f\in F$---except for those $F_2$-edges that picked option (i)---we have additional up to $\xi$ routing queries.
Each such additional query is of one of four types: $(e^{(1)},e^{(2)})$, with $e^{(1)},e^{(2)}\in\{e',e''\}$. 
\end{enumerate}

Overall, this gives up to $r := |F_1| + 2|F_2| + \xi|F| = \OO(\poly(k))$ routing queries.

The number of choices for such a set of routing queries is at most
$2^{|F_1|} \cdot 5^{|F_2|} \cdot (4^{\xi})^{|F|} = \OO(5^{\poly(k)})$.
However, up to now we did not consider a crucial interplay of these individual routing queries: We need to take all possible orderings 
of the edges emanating from a side of $e$ into account: Sides of $e$ arise at most $2r$ times over all queries, and 
we hence have at most $(2r)!$ orderings to consider. Thus, we overall obtain $\OO(5^{\poly(k)}\cdot \poly(k)!) = \OO(\poly(k)!)$ subproblems.
\end{proof}

\paragraph{Dynamic programming and root node.}

Finally, we have to describe how to use these subproblems to efficiently compute MEI.
The validity of this approach for non-dirty pertinent graphs was already established in~\cite{GMW05}.
As mentioned, we consider dirty nodes bottom-up.

Let $\nu$ be the considered SPR-tree node with skeleton $S_\nu$. Let $e_\varrho\in E(S_\nu)$ be the virtual edge corresponding to $\nu$'s father $\varrho$ (if it exists),
and $e_1,\ldots,e_\ell$ ($e'_1,\ldots,e'_{\ell'}$) the dirty (non-dirty) virtual edges in $S_\nu$ corresponding to the children $\mu_1,\ldots,\mu_\ell$ ($\mu'_1,\ldots,\mu'_{\ell'}$, respectively).
We need to show that we can solve each subproblem at $\nu$ purely using $S_\nu$ and the solutions to the subproblems
of the dirty children. In particular, we may not expand the skeleton to the pertinent graph (for which the $\nu$-subproblems are actually defined).

\medskip
\noindent\emph{Subproblems, embeddings, and the root.}
Assume, $\nu$ is a non-root node, then we have to solve $\chi:=\OO(\poly(k)!)$ many subproblems.
For each subproblem, we are given a set of $r:=\OO(\poly(k))$ routing queries, and want to find the optimal
solution over all embeddings of $P_+$. As a first step, we recall that there are only
a bounded number of embeddings for $S_\nu$ ($1$, $2$, and $(2k)!$ in case of an S-, R-, and P-node; let $\chi'=\OO((2k)!)$), and we may
hence enumerate each one explicitly.

The routing queries of the considered subproblem give rise to the following gadget:
Set the weight of $e_\varrho$ to $\infty$, and introduce two edges $e'$ and $e''$ parallel to $e_\varrho$, one directly to its left, 
one directly to its right. Now subdivide these to edges such that there is a vertex on $e'$ ($e''$)
if a routing query specifies the edge side $e'$ ($e''$). Furthermore, these vertices are ordered according to the specification
of the subproblem. Let $S'_\nu$ denote the embedded graph arising from this construction; we do not consider $e_\varrho$ as a virtual edge in the following any more.
Instead of considering the original new edge set $F$, we will now consider
the routing queries $(s,t)$ as edges $st\in F'$ (a new set $F'$) to be inserted.

If $\nu$ is the root node, we also have to consider all its skeleton's possible embeddings $S'_\nu$ individually, but there is no specific subproblem
to consider and we simply set $F':=F$ without any gadget construction.
From now on, we do the same steps, independent on whether considering the root node, or a specific subproblem at a non-root node.

\medskip
\noindent\emph{Virtual edges.}
For each non-dirty virtual edge $e'_i=ab$, $1\leq i\leq \ell'$, we set the weight of $e'_i$ to the minimum-$ab$-cut in 
the pertinent graph of $\mu'_i$. Note that these values can be constructed bottom up in overall linear time as a preprocessing.

Now, for each dirty virtual edge $e_i$, $1\leq i\leq\ell$ ($\ell\leq 2k$), we construct a gadget analogous to the gadget for $e_\varrho$:
Edge $e_i$ gets weight $\infty$, we add two edges $e'_i,e''_i$ left and right of $e_i$, and subdivide them according to a subproblem
at $\mu_i$.---To do this, we have to enumerate all possible choices of subproblems at all virtual edges. So
this construction yields $\chi'':=\OO((\poly(k)!)^{2k})$ different choices, each of which we consider individually. 
Observe: If, for some edge $f'\in F$ that resides within some $P_{\mu_j}$, we chose a routing query of type (b)(i), we do not 
consider subsolutions at any virtual edge where there are type (c) queries w.r.t.~$f'$. We call such an edge $f'$ a \emph{suppressed edge}.

We denote the so-modified plane graph by $S''_\mu$, and 
now have to decide what happens to our new edges $F'$. 
Each edge in $F'$ corresponds to some edge in $F$. Furthermore, we add each non-suppressed edge $f\in F$ to $F'$ if
it has no corresponding edge in $F'$.
We observe that for each vertex $w\in P_\nu\setminus S_\nu$ that is an end in $F'$, there is a unique
\emph{replacement vertex} $r(w)$ in $S_\nu$---it arises from a query $(r,r(w))$ (unoriented) within a subproblem at some dirty virtual edge.

For each original edge $f=uv\in F$, we hence get a partial order of routing queries corresponding to it:
either $u$ ($v$) or its replacement vertex $r(u)$ ($r(v)$, respectively) is in $S'_\nu$, so we start (end) there.
There may or may not be a routing query starting at $u$ (ending at $v$), which we would update to use $r(u)$ ($r(v)$) instead of $u$ ($v$).
Now, between this start and end, $f$ may have to ``visit'' former queries of type~(c) (whose ends are now represented by subdivision 
vertices at edges $e'_j,e''_j$, $1\leq j\leq\ell$). While these former queries are
totally ordered for each individual dirty virtual edge, it is unclear in which order $f$ visits
the different virtual edges. We will enumerate all possible orders to visit each of the $\OO(k)$
dirty virtual edges up to $\OO(\xi)$ times; there are hence $\chi''':=O\left(\frac{(k\,\xi)!}{k\,\xi!}\right)=\OO(\poly(k)!)$
different visit orderings for each edge of $f\in F$. Every visit order induces an unambiguous set $T_f$ of (new) routing queries 
to draw part of $f$ within $S_\nu$: from $f$'s start to the
vertex representing the beginning of a former query, from the vertex representing the end of the last former query to the
beginning of the next former query, and so on, until finally from the vertex representing the end of the last former query
to $f$'s end.
Such a set $T_f$ hence has size at most $\OO(k\xi)=\OO(\poly(k))$.

So, finally, we obtain an instance $\rMEI(S'_\nu,F'')$, where $|V(S'_\nu)| = \OO(|S_\nu|\cdot k\xi) = \OO(|S_\nu|\cdot\poly(k))$ and $F''$ is the set of all routing
queries (interpreted as unordered new edges) obtained from $F'$ by considering each $T_f$ (for all $f\in F$).
We have $|F''|=\OO(k\,\poly(k)) = \OO(\poly(k))$.
The total cost of the considered subsolution (and also for the solution at the root node) is the minimum number of crossings over all possible r-MEI instances constructed as above \emph{plus}
the numbers of crossings given by the corresponding individual subsolutions realized at the dirty virtual edges.
We have:
\begin{lem}
 We settle the root node---and any specific subproblem at a non-root node---with $\OO(\chi'\cdot \chi''\cdot \chi'''^k)$ calls to r-MEI. We settle each dirty non-root node
 with $\OO(\chi\cdot \chi'\cdot \chi''\cdot \chi'''^k)$ calls to r-MEI.
\end{lem}

\begin{thm}[Detailed version of Theorem~\ref{thm:bicon}]
 Let $G$ be a planar biconnected graph on $n$ vertices, and $F$ a set of $k$ new edges (vertex 
pairs, in fact) where $k$ is a constant.
We can solve $\MEI(G,F)$ in 
$\OO(n\cdot (\poly(k)!)^{\Theta(k)})=\OO(n)\cdot k^{k^{\OO(1)}}$ time. 
\end{thm}
\begin{proof}
First, due to Theorem~\ref{thm:rigidalg}, each individual r-MEI instance in our setting can be computed 
within $\OO(|V(S_\nu)|\cdot\poly(k)\cdot 2^{\poly(k)})=\OO(|V(S_\nu)|\cdot 2^{\poly(k)})$ time. Furthermore,
the union over all SPR-tree skeletons has still linear size $\OO(n)$. We hence obtain the overall runtime
\begin{align*}
\OO(\chi\cdot \chi'\cdot \chi''\cdot \chi'''^k\;\cdot\; n\cdot2^{\poly(k)}) &= \OO(n\cdot \poly(k)! \cdot (2k)! \cdot (\poly(k)!)^{2k} \cdot (\poly(k)!)^k \cdot 2^{\poly(k)}) \\
&= \OO(n\cdot (poly(k)!)^{\Theta(k)})
\end{align*}
\end{proof}

\paragraph{Connected Case.}
Until now, we only considered biconnected $G$. In case of only connected $G$, we can first decompose (in linear time)
$G$ into its biconnected components (blocks), and establish a BC-tree $\scri B$. This tree has two types of nodes: For each block of $G$,
we have a node of type {\bf(B)}; for each cut vertex in $G$, we have a node of type {\bf(C)}. We have an edge $\beta\gamma$ in $\scri B$ if, and only if, 
$\beta$ is a B-node, $\gamma$ is a C-node, and the block of $\beta$ contains the cut vertex of $\gamma$. 
We may root $\scri B$ arbitrarily at any dirty block; we say a block is \emph{dirty} if it contains at least one end of $F$ (other than 
possibly its parent cut vertex). Clearly, we can iteratively prune non-dirty B-leaves.

Now, we can construct a combined tree $\scri C$: For each block $B$ in $G$, we construct (and root) its SPR-tree $\scri T_B$.
In $\scri B$, we replace each B-node with the root vertex of the block's corresponding SPR-tree.
Now, we can run the dynamic programming algorithm over $\scri C$ instead of a single SPR-tree tree.

Let $\nu$ be a non-C-node whose parent is a C-node $\gamma$ corresponding to cut vertex $c\in S_\nu$. 
We need to redefine the subproblems to consider at $\nu$: instead of considering routing queries that attach to one of the two sides of the parent virtual edge,
our routing queries may now attach to $c$ in a specified order and through specified faces incident to $c$.
We therefore introduce the gadget---for each considered embedding $S'_\nu$ of $S_\nu$---obtained by planarly replacing $c$ by a simple cycle $C$. The $c$-incident edges
are attached to $C$ such that the contraction of $C$ again gives $S'_\nu$. 
When considering the routing queries, instead of the two choices of the side of the parent virtual edge, we now hence have a $\delta(c)=\OO(\Delta_\mathrm{cut})$-fold choice over the
segment of $C$ where to attach to, where $\Delta_\mathrm{cut}$ denotes the maximum degree over all cut vertices.

In our dynamic programming, we will perform no operation at C-nodes, but let $\nu$ now be a node with a C-child $\gamma$ corresponding to cut vertex $c\in S_\nu$.
Analogous to above---in each considered embedding $S'_\nu$ of $S_\nu$---we planarly replace $c$ by a cycles $C$. On $C$, we realize
all subsolutions of all (at most $2k$) children of $\gamma$, in all possible combinations.
Except for these modifications, the algorithm remains unchanged, and we obtain:
\begin{thm}[Detailed version of Theorem~\ref{thm:connected}]
Let $G$ be a planar connected graph on $n$ vertices, and $F$ a set of $k$ new edges (vertex 
pairs, in fact), where $k$ and $\Delta_\mathrm{cut}$--- the maximum degree of the cut vertices in $G$---are constant.
We can solve $\MEI(G,F)$ in 
$\OO(n\cdot {\Delta_\mathrm{cut}}^{\Theta(k)} \cdot (\poly(k)!)^{\Theta(k)})$ time.
\end{thm}

\end{accumulate}

\paragraph{Acknowledgments.}
We thank Sergio Cabello and Carsten Gutwenger for helpful discussions.

\bibliographystyle{abbrv}
\bibliography{exactmei}

\ifaccumulating{%
\newpage\appendix\sloppy
\section{Appendix}
\def\thethm{A-\arabic{thm}}\setcounter{thm}{0}
\accuprint
}


\end{document}